\documentclass[10pt, nofootinbib, showpacs,letterpaper,aps,pra,twocolumn,superscriptaddress,floatfix]{revtex4-2}
\pdfoutput=1
\usepackage{revsymb, amsmath, amsfonts, amssymb, enumerate, amsthm, graphicx, braket, relsize, bbm, mathrsfs, mathtools}
\usepackage{adjustbox}
\usepackage[normalem]{ulem}

\usepackage{graphicx,color}
\usepackage{paralist}

\usepackage[utf8]{inputenc}
\usepackage{xfrac}
\usepackage{tabularx}
\usepackage{array}
\usepackage{thm-restate}
\UseRawInputEncoding

\newtheorem{theo}{Theorem}

\newtheorem{lem}[theo]{Lemma}
\newtheorem{cor}[theo]{Corollary}

\newtheorem{quest}{Question}
\newtheorem{resu}{Result}
\newtheorem{nresu}{Numerical Result}

\newcommand{\id}{\mathbb{I}}

\newcommand{\tr}[2]{\mathrm{tr}_{#2} \left\{ #1 \right\}}
\newcommand{\Tr}[1]{\mathrm{tr}\left\{#1 \right\}}

\usepackage{soul}

\usepackage[linktocpage=true, colorlinks=true, linkcolor=blue, urlcolor=blue, citecolor=blue]{hyperref}
\begin{document}

	\title{Typicality of Contextuality}
	
	\author{Vinicius P. Rossi}
	\email{vinicius.pretti-rossi@ug.edu.pl}
	\affiliation{International Centre for Theory of Quantum Technologies, University of  Gda{\'n}sk, 80-309 Gda{\'n}sk, Poland}

	\author{Beata Zjawin}
	\affiliation{International Centre for Theory of Quantum Technologies, University of  Gda{\'n}sk, 80-309 Gda{\'n}sk, Poland}
	
	\author{Roberto D. Baldij{\~a}o}
	\affiliation{International Centre for Theory of Quantum Technologies, University of  Gda{\'n}sk, 80-309 Gda{\'n}sk, Poland}
	\affiliation{Perimeter Institute for Theoretical Physics, 31 Caroline Street North, Waterloo, Ontario N2L 2Y5, Canada}
	
	\author{David Schmid}
	\affiliation{Perimeter Institute for Theoretical Physics, 31 Caroline Street North, Waterloo, Ontario N2L 2Y5, Canada}
	
	\author{John H. Selby}
	\affiliation{International Centre for Theory of Quantum Technologies, University of  Gda{\'n}sk, 80-309 Gda{\'n}sk, Poland}
	
	\author{Ana Bel\'en Sainz}
	\affiliation{International Centre for Theory of Quantum Technologies, University of  Gda{\'n}sk, 80-309 Gda{\'n}sk, Poland}
	\affiliation{Basic Research Community for Physics e.V., Germany}

	\begin{abstract}
		Identifying when observed statistics cannot be explained by any reasonable classical model is a central problem in quantum foundations. A principled and universally applicable approach to defining and identifying nonclassicality is given by the notion of generalized noncontextuality. Here, we study the {\em typicality} of contextuality — namely, the likelihood that randomly chosen quantum preparations and measurements produce nonclassical statistics. Using numerical linear programs to test for the existence of a generalized-noncontextual model, we find that contextuality is fairly common: even in experiments with only a modest number of random preparations and measurements, contextuality arises with probability over $99\%$. We also show that while typicality of contextuality decreases as the  purity (sharpness) of the preparations (measurements) decreases, this dependence is not especially pronounced, so contextuality is fairly typical even in settings with realistic noise.  Finally, we show that although nonzero contextuality is quite typical, quantitatively high {\em degrees} of contextuality are not as typical, and so large quantum advantages (like for parity-oblivious multiplexing, which we take as a case study) are not as typical.  We provide an open-source toolbox that outputs the typicality of contextuality as a function of  tunable parameters (such as lower and upper bounds on purity and other constraints on states and measurements). This toolbox can inform the design of experiments that achieve the desired typicality of contextuality for specified experimental constraints.
		
	\end{abstract}

	\maketitle

	\section{Introduction}

	Certain quantum experiments exhibit correlations that are incompatible with any classical explanation \cite{bell1964on,kochen1967problem,spekkens2005contextuality,bong2020strong}. Deciding whether a given experiment is genuinely quantum in this sense is a central question in the foundations of quantum mechanics. Among the different rigorous approaches to formalizing nonclassicality, the framework of generalized noncontextuality~\cite{spekkens2005contextuality} comes forward as one of the most well-motivated and widely applicable. In this framework, an operational behavior is classically explainable if it can be reproduced by a generalized-noncontextual ontological model, and ruling out the possibility of any such model provides a certification of nonclassicality. 
	
	Contextuality is important both for its foundational implications and for its role as a resource for quantum technologies. Indeed, the fact that quantum theory features contextuality enables many of the quantum advantages over classical communication and information processing tasks. More specifically, generalized contextuality is known to provide the advantage for tasks concerning communication~\cite{spekkens09,saha2019,sumit23, roch22,fonseca25}, computation~\cite{schmid22}, machine learning~\cite{bowles23}, information processing~\cite{spekkens09, ambainis19, chailloux16, yadavalli22}, metrology~\cite{lostaglio20}, state-dependent cloning~\cite{lostaglio22}, and state discrimination~\cite{schmid2018contextual,shin21,flatt22,mukherjee22}. Furthermore, many instances of nonclassicality have been proven to be closely related to generalized contextuality, such as the notions of nonclassicality arising in the study of Bell nonlocality~\cite{SEER,schmid18ineq,wright2023invertible}, the detection of anomalous weak values~\cite{kunjwal2019anomalous}, the observation of anomalous heat flow~\cite{naim24}, and the non-existence of a quasi-probability representation in quantum optics~\cite{NEGATIVITY,Schmid2024structuretheorem}. 
	
	Foundationally, noncontextuality can be motivated as a notion of classical explainability by a number of complementary perspectives. A first motivation leverages a version of Leibniz's principle (see Sec.~\ref{se:prelim}), which has historically seen great success in the construction of compelling physical theories~\cite{spekkens2019ontological}. Another important motivation is the equivalence of this notion with the possibility of a simplex-embedding in generalized probabilistic theories~\cite{schmid2021characterisation}. Moreover, noncontextuality can also be shown to encompass previous notions of classicality, such as that emerging through quantum Darwinism processes~\cite{baldi21} or within macroscopic realism~\cite{schmidmacro}.
	
	A natural question is: how often does contextuality ``accidentally'' arise in an experiment? Could it be that nearly any experimental quantum setup is hiding some sort of nonclassical resource? When studying nonclassicality in and of itself, one usually engineers a specific set of measurements that need to be performed on some carefully chosen set of state preparations. However, even in experiments that are not deliberately chosen to test for contextuality, one can ask whether contextuality arises. Moreover, real experiments may have noise, imperfect control, lack of shared reference frames between the preparation stage of the experiment and the measurement stage of the experiment, or other limitations. How likely is one to generate contextuality nonetheless in such cases? Here we will study these questions by performing analytical and numerical explorations of contextuality scenarios where states and/or measurements are chosen randomly.
	
	More specifically, we aim to answer the following questions:
	
	\begin{quest}\label{q:1}
		
		How typical is it to stumble upon contextuality when we randomly sample a finite number of pure states and projective measurements? How does this typicality change with the number of states and the number of measurements? 
		
	\end{quest}
	
	\begin{quest}\label{q:2}
		
		How does the typicality of contextuality depend on whether the states are pure or not and on whether the measurements are projective or not?
		
	\end{quest}
	
	\begin{quest}\label{q:3}
		
		If one typically stumbles upon contextuality, does one also typically stumble upon a useful resource for quantum advantage?
		
	\end{quest}
	
	To explore the typicality of contextuality, we employ linear programming methods in order to test whether randomly sampled states and effects admit of a (generalized) noncontextual model \cite{selby2024linear}. This approach allows us to computationally determine how often nonclassicality occurs under different conditions. Our work introduces novel ways to investigate contextuality in prepare-and-measure scenarios and complements recent analytical developments that focus on continuous sets of preparations and measurements~\cite{zhang2025reassessing,zhang2025quantifiers}. Our findings shed light on how common contextuality is in quantum systems and provide a quantitative basis for understanding its typical features. Moreover, our implementation can easily be adapted to describe specific experimental setups. Our results and method have the potential to become a tool that informs scientists on how to best design their experimental setups to potentially feature contextuality.

	\subsection{Preliminaries}\label{se:prelim}
	
	The notion of contextuality was first formalized by Kochen and Specker~\cite{kochen1967problem}, to capture the quantum feature that propositions in a quantum world are not necessarily simultaneously decidable—you can write down a set of questions such that you can answer each two of them simultaneously, but you cannot answer all of them simultaneously. From a physicist's perspective, this broadly meant that for a classical hidden variable model to explain certain quantum statistics, the outcome probability would have to depend on the measurement context. It was later shown that the celebrated Bell's theorem~\cite{bell1964on} can be viewed as a special case of Kochen-Specker contextuality~\cite{Acn2015}.  Generalized contextuality, introduced by Spekkens, further elaborates on this idea by formalizing the broader concept of generalized-noncontextual ontological models, which relies on Leibniz's principle of the \textit{identity of indiscernibles}~\cite{spekkens2019ontological}. By this principle, any attempt to provide a further (classical) explanation of quantum theory should consider indistinguishable quantum protocols as being fundamentally equivalent. Spekkens has shown that any ontological models for quantum theory satisfying this assumption will eventually fail to reproduce quantum predictions, deeming quantum theory as contextual.
	
	An operational theory~\cite{hardy2001quantum,barrett2007information,dariano2017} is a framework that describes a physical theory based on laboratory elements so that all the elements in the theory have a clear mapping onto procedures and concepts connected to experimental practice, without an \emph{a priori} demand for an ontological explanation. For instance, in a prepare-and-measure scenario, an operational theory provides the set of all possible ways a system can be prepared ($\mathcal{P}$), all possible ways in which its physical properties can be measured ($\mathcal{M}$), and the outcome labels for each measurement ($K$). Moreover, an operational theory provides a probability rule $p$ for predicting the statistics of any experiment (in our case, prepare-and-measure setups), given by the conditional probability distribution $p(k|M,P)$ for a preparation $P\in\mathcal{P}$ and a measurement outcome $k|M\in K\times\mathcal{M}$. Finally, an operational theory comes equipped with an equivalence relation (here denoted by $\sim$), deeming procedures equivalent if they yield the same values in the probability rule in any possible experiment.  For example, 
	\begin{equation}
		P\sim P'\iff p(k|M,P)=p(k|M,P'),\quad\forall \, k|M\in K\times\mathcal{M};
	\end{equation}
	\begin{equation}
		k|M\sim k'|M'\iff p(k|M,P)=p(k'|M',P),\quad\forall \, P\in\mathcal{P}.
	\end{equation}

	Quantum theory can be recast as an operational theory, in which equivalence classes of preparation procedures are given by density operators, equivalence classes of measurement outcomes by POVM elements, and the probability rule is given by the Born rule. The equivalence relation is naturally captured by the linearity of the space of Hermitian operators.
	
	Another ingredient that has been recently identified as crucial for defining a noncontextual ontological model is that of \emph{diagram preservation}~\cite{Schmid2024structuretheorem,schmid2021unscrambling}. The main idea captured by this notion is that in the description of an experiment, there is a specification of its causal structure--- i.e., what variables can influence which others. Diagram preservation requires taking this structure as something fundamental and imposes that the compositional structure of the ontological model be the same as the causal structure (e.g., the quantum circuit). For example, if the experiment is a Bell scenario, diagram preservation will demand that the response functions for the measurements factorize into one for Alice and one for Bob.
	
	In the prepare-and-measure scenario (which we study in this work), a diagram-preserving ontological model consists of a measurable space $\Lambda$, conditional probabilities $\mu_P(\lambda):\Lambda\to[0,1]$, and response functions $\xi_{k|M}(\lambda):\Lambda\to[0,1]$ with $\sum_k\xi_{k|M}(\lambda)=1$, $\forall\lambda\in\Lambda$ and $M\in\mathcal{M}$, such that
	\begin{equation}
		p(k|M,P)=\sum_{\lambda\in\Lambda}\mu_P(\lambda)\xi_{k|M}(\lambda).
	\end{equation}
	An ontological model in a prepare-and-measure scenario will be a \emph{noncontextual} ontological model when operational equivalences translate into ontological identities. For instance, 
	\begin{equation}
		P\sim P'\implies \mu_P(\lambda)=\mu_{P'}(\lambda),\quad\forall\lambda\in\Lambda;
	\end{equation}
	\begin{equation}
		k|M\sim k'|M'\implies \xi_{k|M}(\lambda)=\xi_{k'|M' }(\lambda),\quad\forall\lambda\in\Lambda.
	\end{equation}
	
	\bigskip
	
	Given a set of preparations and measurements in a prepare-and-measure scenario and the equivalence relation in the operational theory thereof, then, the question is whether the conditional probability distribution $\{p(k|M,P)\}_{P\in\mathcal{P},k|M\in K\times\mathcal{M}}$ admits of a noncontextual ontological model. This question can be systematically tackled by a linear program \cite{selby2024linear} that leverages the geometrical underpinning of generalized contextuality in terms of simplex embeddings \cite{schmid2021characterisation}. As input to the linear program, one provides the density matrices $\{\rho_P\}_{P\in\mathcal{P}}$ that correspond to the state of the quantum system for the different preparations $P\in\mathcal{P}$, and the positive-semidefinite matrices $\{E_{k|M}\}_{k|M\in K\times\mathcal{M}}$ that correspond to the operators for the different measurement outcomes. The program then outputs a decision on whether the preparations and measurements are `simplex embeddable', i.e., whether the statistics $\{p(k|M,P)\}_{P\in\mathcal{P},k|M\in K\times\mathcal{M}}$ admit of a noncontextual ontological model. 
	
	In this work, we randomly sample prepare-and-measure scenarios with a fixed number of preparations and measurements, drawing states and effects from specified distributions. For each scenario, we run the linear program mentioned above to certify nonclassicality. By repeating this procedure many times, we can check how often nonclassicality arises as a function of the scenario parameters, such as the number of preparations (or measurements) and the purity (or projectivity) of these.
	
	\bigskip
	
	\subsection{Methods}\label{sec:methods}
	
	As previously stated, the linear program provided in Ref.~\cite{selby2024linear} assesses the contextuality of a prepare-and-measure scenario by constructing the associated generalized probabilistic theory (GPT) fragment and assessing whether this fragment admits of a simplex embedding. The program takes in a set of density operators $\Omega=\{\rho_P\}_{P\in\mathcal{P}}$ and a set of POVM elements $\mathcal{E}=\{E_{k|M}\}_{k|M\in K\times\mathcal{M}}$ and assesses how much partially depolarizing noise $0\leq r\leq1$ must act on the states of the scenario until a noncontextual model of it becomes possible. The quantity $r$, therefore, is a witness of nonclassicality, since any scenario with $r>0$ is necessarily contextual. This quantity has been previously named \emph{robustness of contextuality}. A summary of the linear program and its implementation\footnote{Available in Mathematica at \url{https://github.com/eliewolfe/SimplexEmbedding}~\cite{elie-git} and Python at \url{https://github.com/pjcavalcanti/SimplexEmbeddingGPT}~\cite{cavalcanti-git}, with the latter being employed in this work.} is given in Appendix~\ref{app:LP}.
	
	To study the typicality of quantum contextuality, we therefore perform extensive simulations of contextuality certification. We characterize each simulation by a tuple of parameters $(n,m,d; N)$, where $n$ is the number of preparations and $m$ is the number of binary measurements that comprise the prepare-and-measure scenario, $d$ is the dimension of the Hilbert space associated to the quantum system being investigated, and $N$ is the number of sampling iterations and assessments of contextuality. We begin by randomly sampling $n$ density matrices of dimension $d \times d$ and $m$ pairs of POVM elements, each of which rescaled to ensure normalization (amounting to a total of $2m$ effects). We then employ the linear program~\cite{selby2024linear} to determine whether the statistics specified by $(n,m,d)$ admit a generalized-noncontextual model. 
	
	This procedure is repeated $N$ times, with new random samples of states and effects in each run. We define the \emph{typicality of contextuality}, denoted by $t_{(n,m,d;N)}$, for a given scenario  $(n,m,d;N)$  as the percentage of samples (out of $N$) that exhibit contextual behavior, i.e., for which the linear program assigned non-zero robustness $r$. When referring to the analytical typicality of contextuality that would emerge out of an infinite number of trials, we adopt $t_{(n,m,d)}$, suppressing the label $N$. 
	
	Notice that, throughout this work, stating $t_{(n,m,d;N)}\approx100\%$ does not mean that contextuality will always be found for $n$ preparations and $m$ binary measurements for a quantum system of dimension $d$, but rather that out of $N$ tries, every single one of them exhibited contextuality by our certification criteria. To calculate the confidence level of the true value of $t_{(n,m,d)}$, we use the Wilson score interval~\cite{wilson}, which we explain in detail in Appendix~\ref{app:threshold}. In summary, when we say $t_{(n,m,d;N=10^6)}\approx100\%$ in this work, this criterion ensures that the typicality of contextuality $t_{(n,m,d)}$ is at least 99.999\% with 99\% level of confidence. In this paper, we focus on qubit states $(d=2)$, but the general approach and the simulation toolbox that we share in an online repository are also suitable for higher dimensions.  
	
	To generate random states, we employ the \texttt{qutip} toolbox \cite{lambert2024qutip5quantumtoolbox}. Pure random states distributed with a Haar measure are generated by applying a Haar random unitary to a fixed pure state (\texttt{rand\_ket} function). To generate mixed states, we follow the algorithm designed in Ref.~\cite{Bruzda2009} that generates Ginibre random density operators of fixed rank (\texttt{rand\_dm} function). Since only full-rank density matrices have nonzero measure~\cite{Bengtsson2017}, sampling from mixed states will never lead to pure states. More details on random sampling are provided in Appendix~\ref{app:qutip}. To sample the effects, there are two approaches that we employ depending on the task--- randomly sampling effects or fixing a large number of (somewhat uniform) effects. For the first approach, we employ \texttt{qutip}'s random unitary generator to sample projective measurements by rotating an orthornomal basis; and \texttt{numpy}'s random generator function to sample POVMs by generating sets of Hermitian operators from the Ginibre ensamble, then enforcing measurement normalization~\cite{kukulski2021} (See Appendix~\ref{app:qutip}).
	For the second task, the idea is to fix a large number of effects that are distributed in all directions on the Bloch sphere. We choose the following parametrization:
	\begin{equation}\label{90effects}
		\ket{j,k}=\left(\begin{array}{c}
			\sin\left(\frac{j\pi}{20}\right)\\
			e^{i\frac{k\pi}{5}}\cos\left(\frac{j\pi}{20}\right)
		\end{array}\right)
	\end{equation}
	for $j=\{0,\dots,9\}$ and $k=\{0,\dots,4\}$. The effects built from these states, together with their normalizing counterparts, form $m=92$ distinct projectors. These projectors can be visualized as points distributed symmetrically around the Bloch sphere along distinct inclinations between the north and south poles, with each point representing a unique measurement direction. Including their antipodal points, this gives the full set of 184 distinct effects hereby considered, which we depict in Fig.~\ref{fig:fixedeffects}.  This set of effects is designed to approximate the set of all projective effects. 
	
	\begin{figure}
		\centering
		\includegraphics[width=0.6\linewidth]{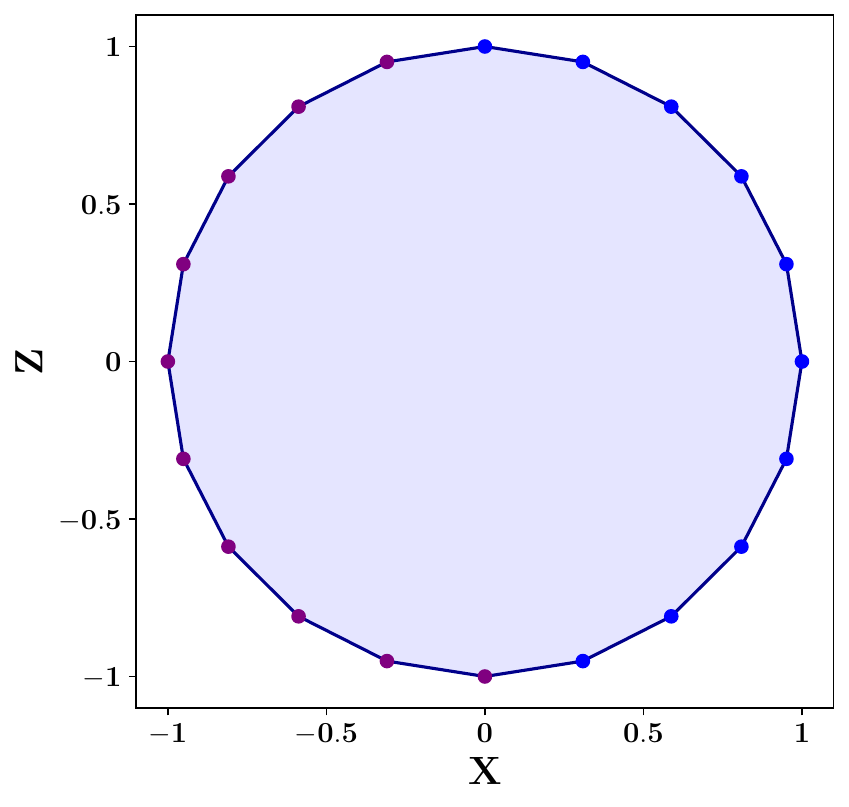}
		\caption{Visualization of the effects $\ket{j,k}$ from Eq.~\ref{90effects} (blue dots) and their antipodal counterparts (purple dots) for $k=0$, showing only the ZX-hemisphere of the Bloch sphere, for clarity. The full set of effects is obtained by additionally rotating this polygon 4 times in steps of $\frac{\pi}{5}$ around the Z axis of the Bloch sphere.}
		\label{fig:fixedeffects}
	\end{figure}
	
	The results of our analysis depend on several technical details, such as choosing the number of runs $N$ in the simulations to be appropriately large. Moreover, for large values of $N$, it is important to specify a suitable threshold for the robustness of contextuality output by the linear program, below which we classify statistics as classical. This is needed since numerical approximations and errors can lead to non-zero robustness even for classical scenarios. This threshold may depend on the scenario, the solver employed by the linear program, and the internal optimization settings. In the following section we will always adopt $N=10^6$ and the threshold to be $10^{-7}$ (i.e., whenever $r\leq 10^{-7}$ the data will be deemed noncontextual), and the investigation that backs these decisions (alongside the analysis of the best solver for the linear program for this task) is provided in Appendix~\ref{app:threshold}. 
	
	\section{Typicality of contextuality}

	\subsection{The impact of the number of states and measurements}
	\label{sec:results1}

	We start this section with a simple analytical result: that if the number of states is too small for the given system's dimension, then an experiment with random sampling will virtually always admit of a classical explanation.
	
	\begin{lem}[\textbf{Typicality zero for scenarios $\mathbf{(n<d^2,m,d)}$ }] \label{lemma:TipiZeroForLI}\quad\\
		\noindent Consider a quantum system of Hilbert space dimension $d$, and a typicality scenario $(n,m,d)$ with $n\leq d^ 2$ quantum states (randomly sampled without any additional restrictions). Then, for any number of measurements $m$, the typicality is $t_{(n,m,d)}=0\%$.
		
	\end{lem}
	
	\begin{proof}
		
		Recall that the space of $d\times d$ complex Hermitian matrices is of dimension $d^2$ and we can see states as vectors in that space. Therefore, if $n\leq d^ 2$, the set of samples such that $\{\rho_i\}_{i=1}^ n$ is linearly dependent is of measure zero. We can then consider a set $\{\rho_i\}_{i=1}^ n$, sampled at random with $n\leq d^ 2$, to be a linearly independent set.  As shown in Ref.~\cite{zhang2025reassessing}, if a set of states is linearly independent, then there are no prepare-and-measure experiments where one can witness contextuality, even if all measurements are considered. In other words, the fragment $(\{\rho_i\}_{i=1}^n,\mathcal{E}^ {\rm all})$ admits of a noncontextual model, where $\mathcal{E}^ {\rm all}$ denotes the set of all effects. This implies that any other fragment of the form $(\{\rho_i\}_{i=1}^n,\mathcal{E}')$, where $\mathcal{E}'\subset{\mathcal{E}^{\rm all}}$, is noncontextual (as simplex-embeddability is transitive~\cite{schmid2021unscrambling}).  
		
	\end{proof}
	
	Lemma~\ref{lemma:TipiZeroForLI} sets a threshold on the number of states: below it, randomly chosen experiments will almost always admit of a classical explanation. This highlights the fact that contextuality tests are usually {\em engineered} to be {\em non}-typical, since contextuality in qubit scenarios with 4 preparations has been extensively studied~\cite{schmid2018contextual,spekkens09,pusey2018,catani2024}. Such studies clearly do not consider randomly chosen states and effects—at a minimum, they choose states that all lie in a plane of the Bloch ball. (Typically, the states and measurements are chosen to satisfy further constraints as well, motivated by some particular physical situation or information processing task being considered.) One may hence ask the question: do we always need to engineer non-typical scenarios if we want to observe contextuality, even for general numbers of preparations and measurements?

	There are reasons to expect that the answer is negative, i.e., that contextual experiments become more typical as we increase the number of states and effects. By having more states (effects), one is likely to probe more of the qubit state (effect) spaces—which are, of course, contextual when considered in their entirety. Moreover, if a set of pure states $\{\ket{\psi_i}\!\bra{\psi_i}\}_{i=1}^n$ is linearly dependent, then the fragment $(\{\ket{\psi_i}\!\bra{\psi_i}\}_{i=1}^n,\mathcal{E}^{\rm all})$ is contextual~\cite{zhang2025reassessing}, which suggests that purity can matter. In the following numerical analyses, we explore more quantitatively how typicality of contextuality depends on these factors. 
	
	We start by exploring how different numbers of states and measurements in a qubit impact the likelihood of witnessing contextuality. Fig.~\ref{fig:3dplots}(a) and Fig.~\ref{fig:3dplots}(b) show the results of simulations characterized by $(n,m,d=2;N=10^6)$ for pure states with projective measurements and mixed states with POVMs, respectively. We consider only scenarios with $n\geq4,m\geq2$,  where these bounds are, on the one hand, motivated\footnote{ Lemma~\ref{lemma:TipiZeroForLI} suggests we could start with $n=5$ states, but we decided to start with $4$ as a way to control the numerical analysis (as we shall see) and also to help visualize the change in typicality as $n$ breaks the threshold.} by Lemma~\ref{lemma:TipiZeroForLI},   and, on the other hand, informed by the fact that a single binary measurement may never reveal contextuality \cite{selby2023contextuality}.  We increase the values of $n$ and $m$ up to 19 (inclusive), which is a considerably large number compared to most prepare-and-measure scenarios that are commonly explored in contextuality studies.
	
	One can see that for $n=4$, typicality of contextuality is given by $t_{(n=4,m,d=2;N=10^6)}=0\%$ for any number of measurements for both pure and mixed states as required by Lemma~\ref{lemma:TipiZeroForLI}. We further discuss this case in Appendix~\ref{app:threshold}, where the choice of solver, robustness threshold, and number of iterations was informed by demanding that  $t_{(n=4,m=2,d=2;N)}=0\%$ for the values of $N$ considered in this work, since deviations from this value strongly suggest issues with numerical approximations or errors. For $n>4$, typicality is non-zero, and it is interesting to observe that it quickly rises to approximately 100\% when the preparations are given by pure states. This suggests that linear dependence of (pure) states is most often a guarantee that the scenario will generate contextual correlations, even if one has access to a finite set of measurements.
	
	Comparing the plots for pure and mixed states, it is immediately apparent that scenarios with pure states and projective measurements require fewer preparations and measurements to achieve high typicality of contextuality and also reach the saturation of $t_{(n,m,d;N)}\approx 100\%$ faster.  For example, for pure states and effects, we observe $t_{(n,m,d;N)}>99\%$ already for $(n=7,m=8)$, while for mixed states and effects values above this range are never achieved for the considered values of $n$ and $m$, and the maximum value registered for this case is $t_{(n,m,d;N)}=97.7\%$. Moreover, notice that if one only has access to $m=4$ measurements, typicality is $t_{(n,m,d;N)}>99\%$ for all $n>10$ pure preparations. Additionally, contrarily to the pure case, the typicality of contextuality never reaches $t_{(n,m,d;N)}\approx100\%$ for the mixed case in this range of numbers of states and measurements, although we do expect that it would if $n$ and $m$ were increased further. This numerical exploration, which pertains to Questions \ref{q:1} and \ref{q:2} in the Introduction, can be summarized as follows. 
	
	\begin{figure*}[htb!]
		\centering
		a)\adjustbox{valign=t}{\includegraphics[width=0.3\linewidth]{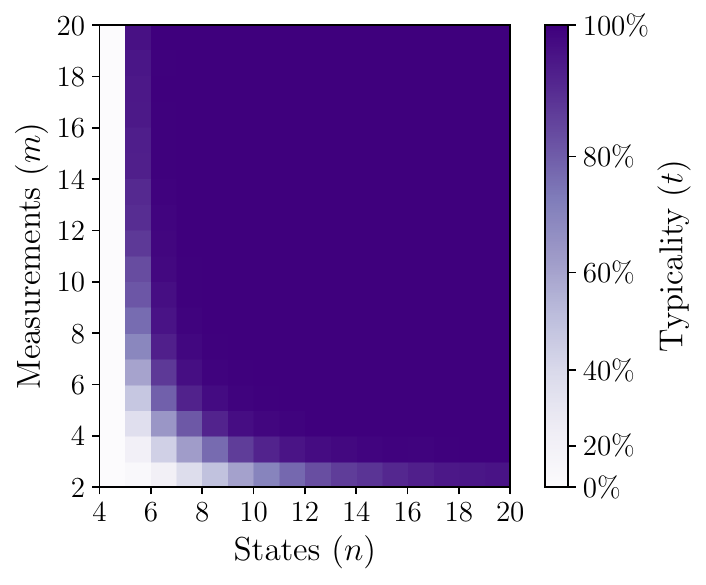} }
		b)\adjustbox{valign=t}{\includegraphics[width=0.3\linewidth]{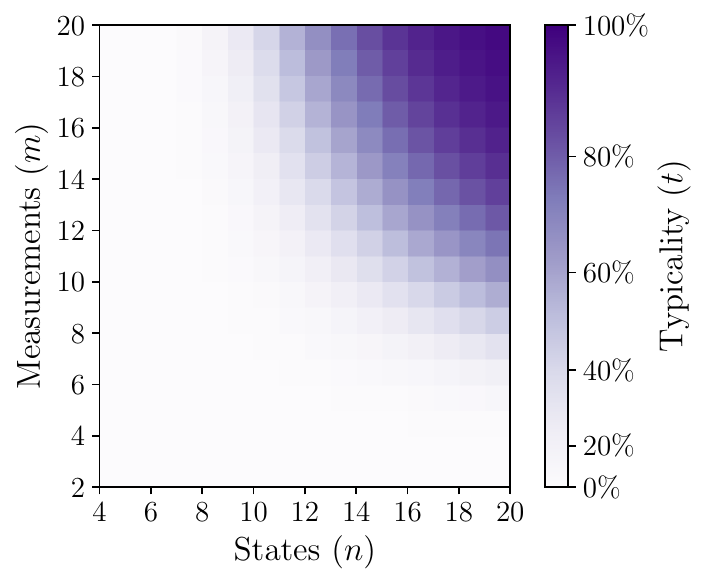}}
		\caption{Typicality of contextuality for different numbers of preparations and measurements for randomly-sampled (a) pure states and projective measurements, and (b) mixed states and POVMs.}
		\label{fig:3dplots}
	\end{figure*}
	
	\begin{nresu}
		For $(n,m,d=2;N=10^6)$, typicality of contextuality is zero when $n=4$ for all $m$, for both (i) pure states and sharp measurements and (ii) mixed states and unsharp measurements. When $n>4$, typicality of contextuality is already non-zero for pure states and sharp measurements, and quickly reaches $t_{(n,m,d;N)} \approx 100\%$ for $m\geq4$. For mixed states and POVMs, a similar behavior is observed, but $t_{(n,m,d;N)} \approx 100\%$ is approached at a slower pace (i.e., larger values of $n$ and $m$ are needed). 
	\end{nresu}
	
	In order to further explore Questions~\ref{q:1} and~\ref{q:2}, we now move on to study how typicality of contextuality assessments change with the number of states when the number of measurements is large and chosen to approximate the full set of projective effects. This would provide a numerical visualization of the result in Ref.~\cite{zhang2025reassessing}, which states that if one has access to all possible measurements for a quantum system, linear dependence among a set of states is necessary (and in the case of pure states, sufficient\footnote{For mixed states, linear dependence is not sufficient to imply contextuality. For example, an $\varepsilon$-ball of states around the maximally mixed state in the qubit space contains as many linear dependences as the unit ball but will nonetheless, for sufficiently small but non-zero $\varepsilon$, admit a noncontextual explanation in any prepare-and-measure scenario (since any contextual prepare-and-measure scenario becomes noncontextual under a finite amount of  depolarizing noise~\cite{schmid2018contextual,marvian2020,ravi2015}).}) to generate contextuality. We start the discussion with an analytical result: 
	
	\begin{restatable}{prop}{upperbound}\label{prop2}
		For any finite $(n,m,d)$ with $n$ preparations and $m$ measurements, $t_{(n,m,d)}<1$. This is holds even when sampling only pure states and effects.
	\end{restatable}
	
	A proof is given in Appendix~\ref{app:proofs}.
	
	This reiterates that assessments of the form $t_{(n,m,d;N)}\approx100\%$ throughout this work do not translate to the theoretical typicality $t_{(n,m,d)}=1$. As we see now, however, with a large but finite number of measurements one consistently gets $t_{(n,m,d;N)}\approx100\%$. Later on, in Prop.~\ref{prop3}, we will show that $\lim_{m\to\infty}t_{(n,m,d=2)}=1$ for $n>4$.
	
	\begin{figure}[t]
		\centering
		\includegraphics[width=0.65\linewidth]{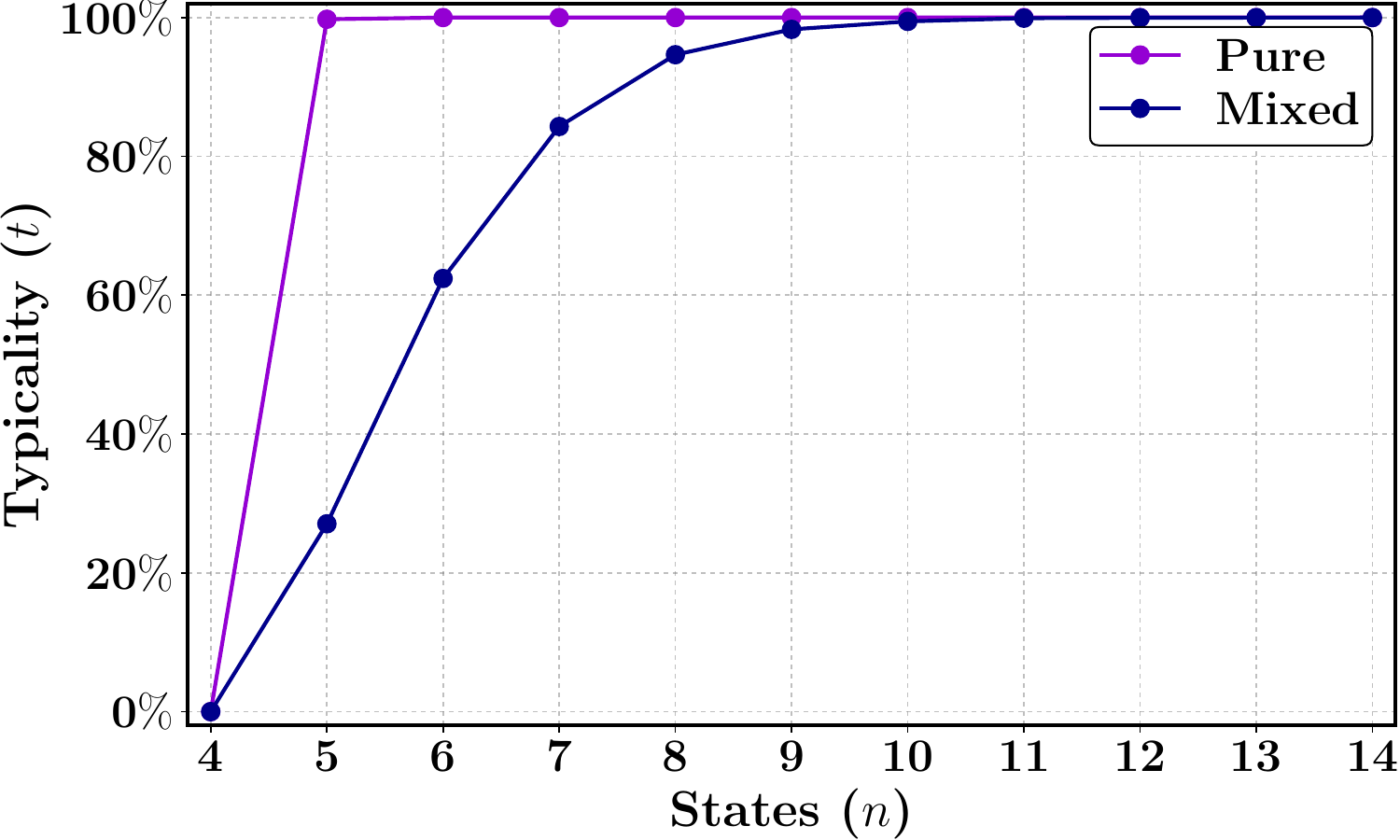}
		\caption{Typicality of contextuality for 92 dichotomic projective measurements over a qubit (that approximate the set of all projective effects), parametrized by Eq.~\ref{90effects}, and $n$ random (pure in violet, mixed in blue) states.}
		\label{fig:manyeffects}
	\end{figure}
	 
	We parametrize projectors of a qubit according to Eq.~\eqref{90effects} and test the typicality for $ (n, m=92, d=2; N=10^6)$ for both pure and mixed states, with the caveat that now the measurements remain fixed for all trials. The results are displayed in Fig.~\ref{fig:manyeffects}. As expected, for $n=4$ we observe $t_{(n,m,d;N)}=0\%$ for both pure and mixed preparations, as generating contextuality in this case is only possible in very specific setups that will never happen in our random sampling. For $n\geq5$, we have non-zero typicality of contextuality, and, in fact, for pure states, we observe $t_{(n,m,d;N)}\approx 100\%$ already for $n \geq 6$. The case with pure states and $n=5$ registers $t_{(n=5,m=92,d=2;N=10^6)}=99.76\%$, and the calculated Wilson score interval gives us that the true value of $t_{(n,m,d)}$ lies above $98.99\%$ with 99\% confidence. As mentioned before, for pure states, if one has access to all possible measurements, linear dependence is both necessary and sufficient to generate contextuality. Here, however, we do not observe a typicality of contextuality $t_{(n=5,m,d;N)} \approx 100\%$, even though five randomly sampled pure qubit states will always be linearly dependent. One could argue that this is due to the fact that, although large, the number of measurements in the experiment is finite. However, whether or not the possible associated contextuality is witnessed by the linear program in Ref.~\cite{selby2024linear} will depend on how dense the measurements are and on the numerical threshold used for the robustness.  Hence, the fact that we have a finite number of measurements is not necessarily the only cause for typicality not being approximately 100\%.
	
	When sampling from mixed states, however, one requires an increasingly large number of states to reach $t_{(n,m,d;N)} \approx 100\%$, and $n=5$ displays a rather low typicality value. A visual comparison of one sampling for the pure and mixed cases for $n=5$ is given in Fig.~\ref{fig:samples1}: there we present (i) the convex hull of the measurements as a blue polytope, which is the same for both the pure state case and the mixed state case, and (ii) the convex hull of the $n=5$ sampled states as a purple polytope, when these are sampled from our methods for pure (Fig.~\ref{fig:samples1}(a)) and mixed states (Fig.~\ref{fig:samples1}(b)). It then appears that for the mixed case, the volume of the states polytope is considerably smaller than for the pure case. Therefore, a reasonable explanation is that the polytope of mixed states can still be easily embedded into a simplex. For the mixed preparations, typicality reaches $t_{(n,m,d;N)}\geq99\%$ for $n=10$, and $t_{(n,m,d;N)} \approx 100\%$ is observed starting at $n=14$. This reflects the fact that linear dependence is not sufficient to guarantee contextuality if one's states are mixed.
	
	\begin{figure*}[htb!]
		\centering
		a)\adjustbox{valign=t}{\includegraphics[width=0.3\linewidth]{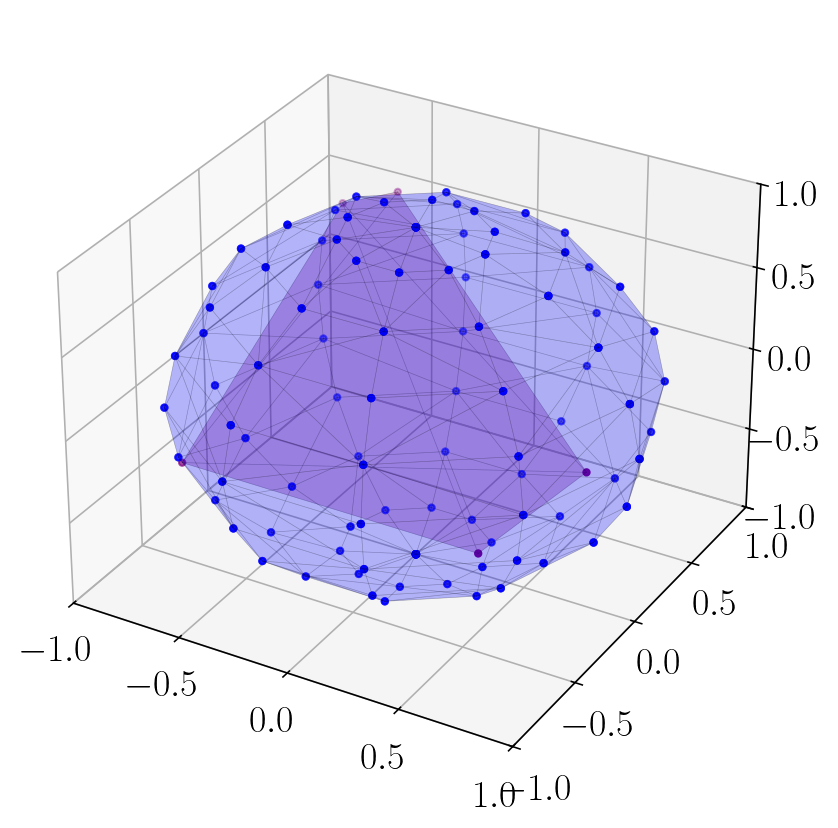}}
		b)\adjustbox{valign=t}{\includegraphics[width=0.3\linewidth]{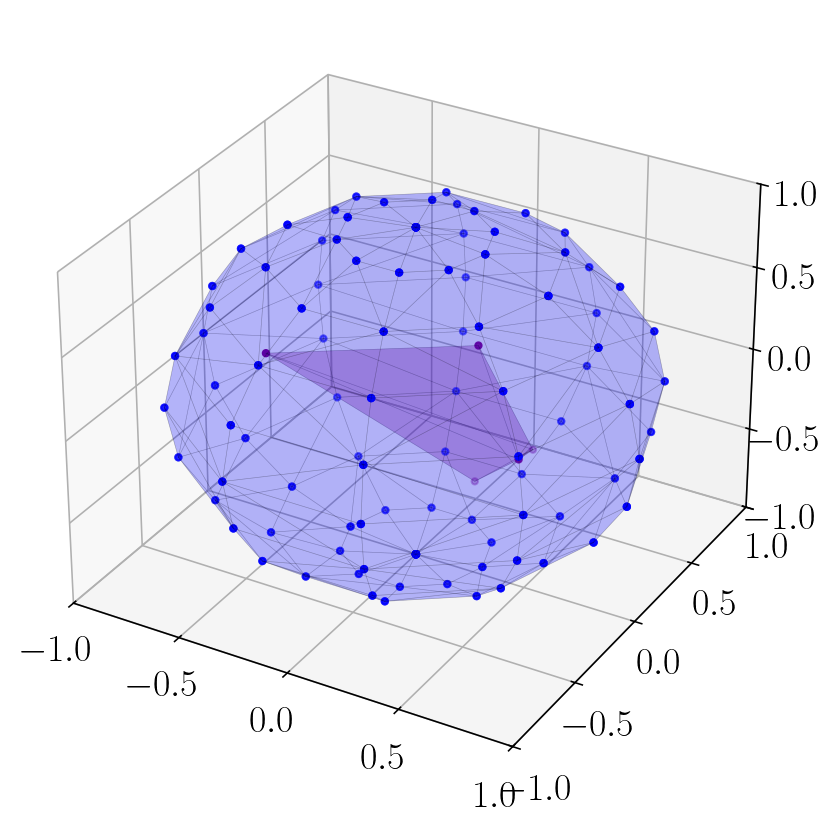}}
		\caption{Examples of $n=5$ (a) pure states and (b) mixed states and their convex hulls (violet), measured by the same 92 projective measurements provided by Eq.~\eqref{90effects} and their convex hull (blue, excluding the null and the unit effects).}
		\label{fig:samples1}
	\end{figure*} 
	
	The numerical results pertaining to a large (but fixed) number of measurements can be summarized as follows:
	
	\begin{nresu}
		
		For a large number of projective measurements on the qubit (here taken to be 92) that roughly approximate the space of all projective effects, the typicality of contextuality is non-zero for $n>4$ (i.e., the range where one could expect this to be the case). For pure states, the typicality of contextuality is quite large for $n=5$, and $t_{(n,m,d;N)} \approx 100\%$ for $n\geq6$. For mixed states, typicality of contextuality is quite large for $n\geq 10$, and $t_{(n,m,d;N)} \approx 100\%$ for $n\geq14$.
		
	\end{nresu}
	
	A useful way of interpreting these results is given by the recent work in Ref.~\cite{zhang2025reassessing,zhang2025quantifiers}. While generalized (non)contextuality is usually viewed as a property to be assessed for a set of observed data or a complete experiment, Refs.~\cite{zhang2025reassessing,zhang2025quantifiers} showed how this idea can be lifted to a notion of nonclassicality that applies even to a single component or fragment of an experiment (such as a single quantum state, a single quantum channel, a set of states, and so on). For example, a set of quantum states is classical if, when paired with all possible quantum measurements, it yields statistics that admit a generalized-noncontextual model. In the result just above, we have considered (a finite approximation of) the entire effect space of quantum theory, and so our result can be interpreted as approximately computing the typicality of nonclassicality for sets of quantum states (as opposed to of entire prepare-measure experiments). In fact, we now prove that in the asymptotic limit, the analytical typicality reaches $\lim_{m\to\infty}t_{(n,m,d)}=1$.
	
	\begin{restatable}{prop}{asymptotic}\label{prop3}
		Consider samples of $n$ pure states and $m$ pure effects of a $d$-dimensional quantum system. For any $n>d^2$ number of randomly sampled pure states, 
		\begin{equation}\label{eq:lim-effects}
			\lim_{m\to\infty}t_{(n,m,d)}=1.
		\end{equation}
		Moreover, for $m > d^2$ randomly sampled pure effects,
		\begin{equation}\label{eq:lim-states}
			\lim_{n\to\infty}t_{(n,m,d)}=1. 
		\end{equation}
	\end{restatable}
	
	The proof is given in Appendix~\ref{app:proofs}.
	
	The case investigated in this work of sampling dichotomic measurements naturally applies to this proof (in which case $m\geq d^2/2$ already ensures that the scenario is not ruled out by Lemma~\ref{lemma:TipiZeroForLI}). A useful property stemming from this result is that $t_{(n,m,d)}$ grows monotonically with $n$ and $m$. 
	
	\begin{cor}
		For any $n>d^2$ number of randomly sampled pure states and any real number $t<1$, there is a number $m\in\mathbb{N}$ of randomly sampled pure effects such that $t_{(n,m,d)}>t$. Similarly, for any $m> d^2$ randomly sampled pure effects, there is a number $n\in\mathbb{N}$ of randomly sampled pure states such that $t_{(n,m,d)}>t$.
	\end{cor}
	
	We have shown that contextuality is quite typical even in experiments with a small number of pure states and projective measurements, and that typicality quickly approaches 100\% as the number of states or measurements increases. However, real experiments never involve pure states or projective measurements, so the investigation so far is primarily of theoretical interest. In the next section, we consider the robustness of these results when the assumptions of purity and projectivity are relaxed.
	
	\subsection{The impact of purity and sharpness}
	
	Since most experiments have a fairly good understanding of how pure the prepared states are and how sharp the implemented measurements are, we now shift our attention to scenarios in which there are known upper and lower bounds on these quantities. This lets us study how the typicality of contextuality is affected by the purity of the states and projectivity of the measurements.  
	
	We explore this by analyzing scenarios with the fixed number of measurements $m=20$. From the results in the previous section, we know that for any $m \geq 12$, typicality of contextuality reaches $t_{(n,m,d;N)}>99\%$ for $n \geq 7$ pure states. Moreover, for $m \geq 13$, typicality $t_{(n,m,d;N)}> 99\%$ is already achieved for $n=6$. Our current question is: for $m=20$ measurements, how does the minimum number of states required to observe $t_{(n,m,d;N)}> 99\%$ depend on the purity of the states and the set of allowed measurements? To address this, we examine different lower bounds on purity (measured by $\Tr{\rho^2} \in \left[\frac{1}{d},1\right]$, where $\rho$ is the state of the system), thereby restricting how mixed the states can be. We study two cases: (i) restricting the purity of states while keeping the measurements projective and (ii) restricting the purity of both states and the sharpness of measurements. The results for $(n,m=20,d=2;N=10^6)$ are presented in Fig.~\ref{fig:purity} for three lower bounds on purity/sharpness: 0.5 $=\frac{1}{d}$, 0.7, and 0.9, and two upper bounds on purity/sharpness: 0.95 and 1. 
	
	\begin{figure*}[htb!]
		\centering
		\begin{tabular}{cc}
		a)\,\,\adjustbox{valign=t}{\includegraphics[width=0.35\linewidth]{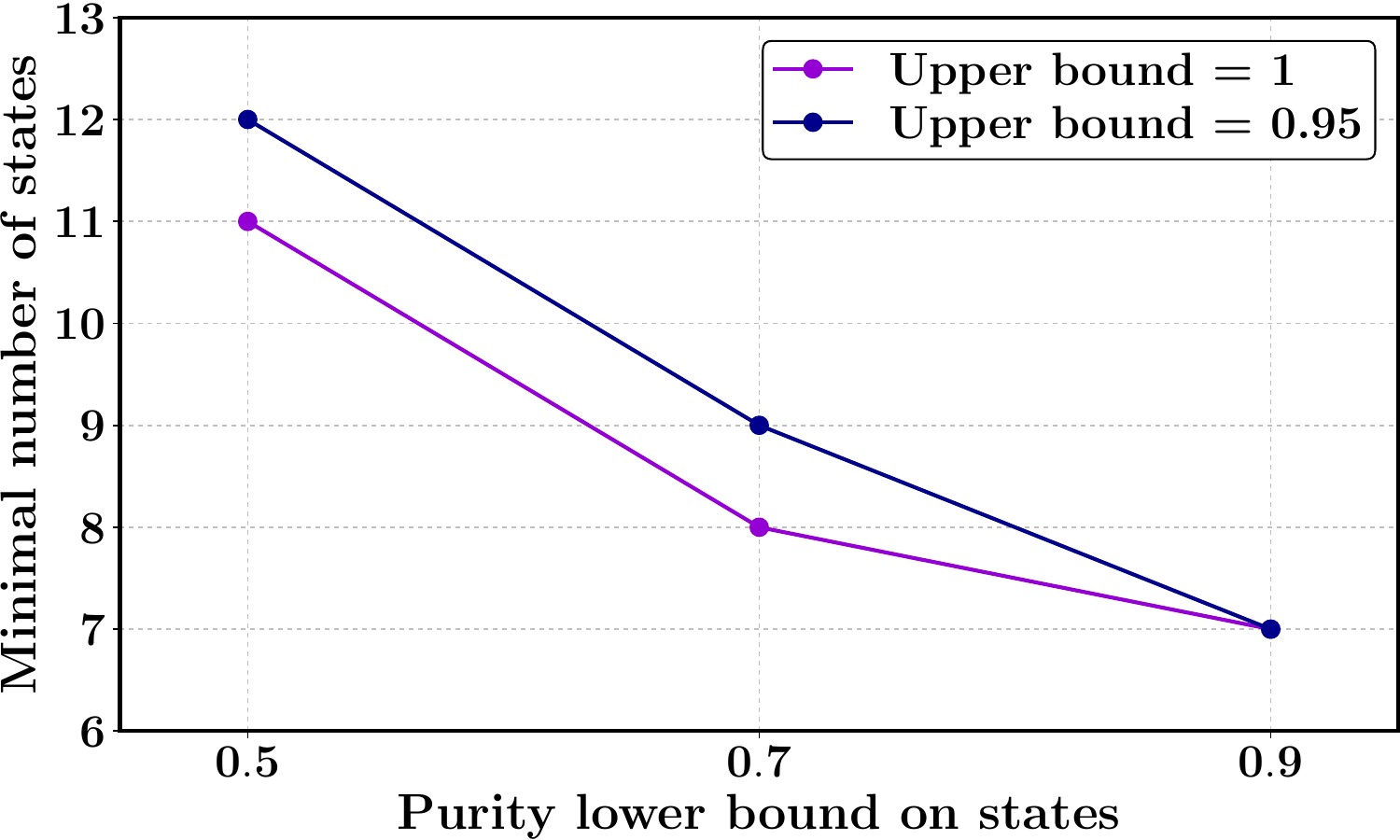} }&
		b)\,\,\adjustbox{valign=t}{\includegraphics[width=0.35\linewidth]{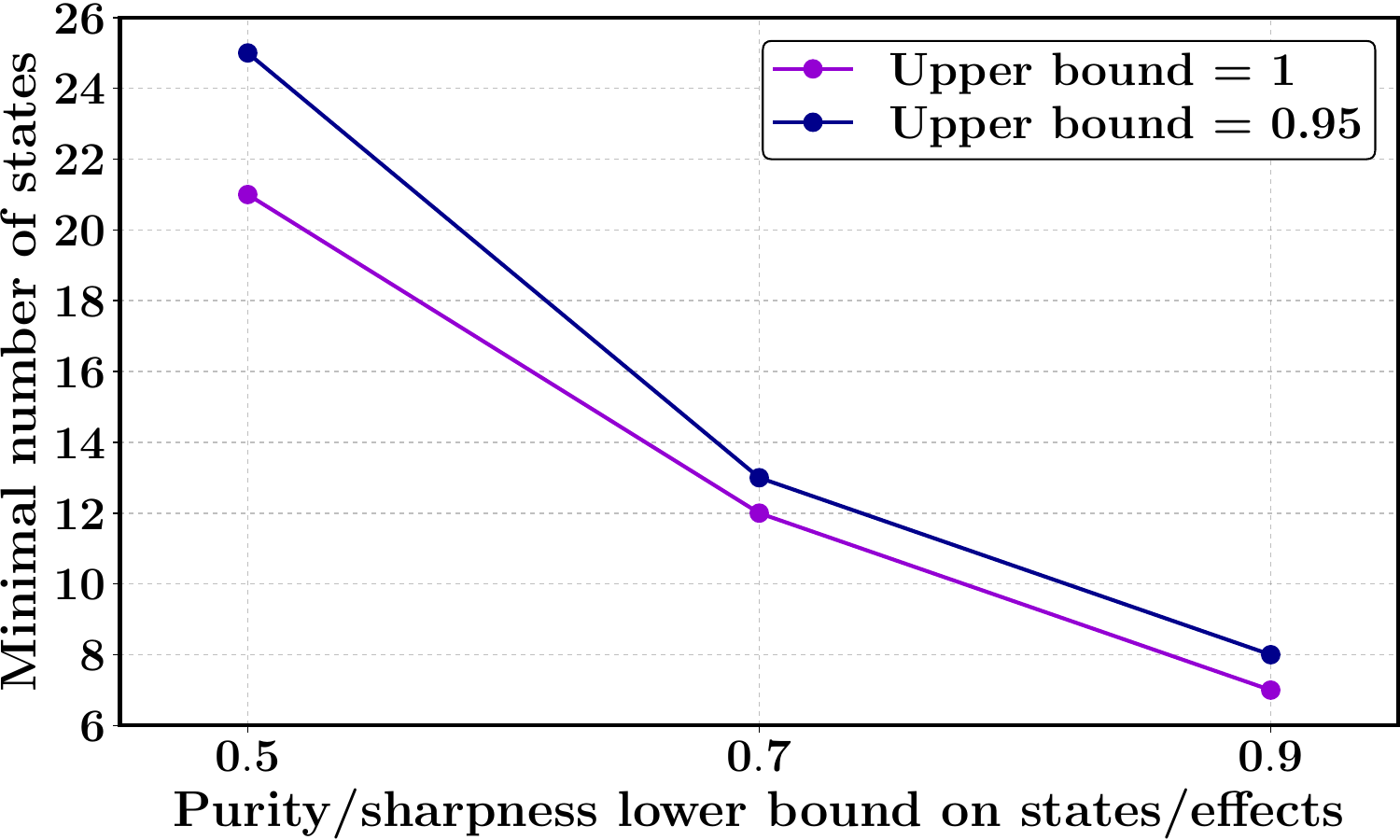}}
		\end{tabular}
		\caption{Minimal number of states $n$ needed to observe typicality $t_{(n,m=20,d=2;N=10^6)} >99\%$ for $m=20$ measurements as a function of the lower bound on purity of the states with (a) sharp measurements and (b) unsharp measurements, for two different upper bounds on their purity. In (b), the bounds on sharpness are the same as the purity of the states.} 
		\label{fig:purity}
	\end{figure*}
	
	For the case of projective measurements (Fig.~\ref{fig:purity}(a)), one requires a considerably lower number of random states (compared to the case of POVMs in Fig.~\ref{fig:purity}(b)) to generate contextuality with typicality $t_{(n,m,d;N)}> 99\%$. From Fig.~\ref{fig:purity}(a), notice that eliminating a thin slice of highly faithful states (i.e., the ones for which purity is between $0.95$ and $1$) does not jeopardize this ubiquity of nonclassicality—at least when the measurements are sharp. In fact, if the states are pure enough (purity larger than $0.9$), it does not matter whether one has access to this highly faithful slice or not. Our results thereby show that for 20 projective measurements, $n=9$ preparations already suffice for near-certain typicality when purity is at least $0.7$. If one can ensure at least $0.9$ purity in states, one requires only $n=7$ preparations. 
	
	Naturally, these numbers may change if the randomly sampled measurements are unsharp. As shown in Fig.~\ref{fig:purity}(b), the minimal number of random preparations required for $t_{(n,m,d;N)}> 99\%$ with 20 random POVMs is as high as $n=25$ for the case where no lower or upper bounds on purity are imposed. Notice that because they are a measure-zero subset of the full effect space, projective effects are never sampled in this case. It is therefore interesting to notice that the minimal number of preparations is $n=7$ for states and effects with purity between $0.9$ and $1$. This is the same number as the one observed for projective measurements and states in this slice, suggesting that, within this range of purity, contextuality is easy to generate even without access to projective measurements. This is extremely relevant for experimental implementations, where high purity states or sharp measurements may be costly or impossible to implement~\cite{Girolami2019,Guryanova2020}.
	
	These numerical results pertaining to constrained-purity scenarios (related to Question \ref{q:2} in the Introduction) can be summarized as follows:
	
	\begin{nresu}
		
		The minimal number of states $n$ that are required to generate contextuality with typicality of contextuality $t_{(n,m,d;N)}>99\%$ for the prepare-and-measure scenario $(n,m=20,d=2)$ increases as the lower bound on the purity of the $n$ states is decreased. This happens both when the measurements are projective and when the measurements have some fixed range of purity. In the latter case, $n$ is generically larger than in the former. For projective measurements, having states with purity in the range $(0.9,0.95)$ yields the same results as having states with purity in the range $(0.9,1)$. A similar behavior happens for non-projective measurements, where the results for both purity ranges are similar. This suggests that it is not necessary to experimentally prepare highly pure states or highly projective measurements in order to typically observe contextuality.  
		
	\end{nresu}

	All the values for the minimal number of states depend on the number of fixed measurements $m$, as well as the other parameters involved in our computations and discussed in detail in Appendix~\ref{app:threshold}. In order to facilitate the investigation for particular scenarios, we provide the open-source Python implementation that generated all of our plots in Ref.~\cite{github}, along with a function that can be used to obtain the minimal number of preparations to be sampled for $t_{(n,m,d;N)}>99\%$ given $m$, $d$, $N$ and upper and lower bounds for the purity of states (and possibly effects).  A summary of this repository is provided in Appendix~\ref{app:repo}. 
	
	\subsection{Comparison to typicality of Bell nonlocality}
	
	Contextuality is closely related to a type nonclassicality, called nonlocality~\cite{bell1964on,Brunner2014}. Nonlocality arises in Bell scenarios, where statistics generated by spacelike-separated measurements on systems prepared in entangled states may sometimes resist a classical explanation.  In Bell scenarios, the outcome statistics are usually referred to as \textit{correlations}.  For two parties (Alice and Bob) in a Bell scenario, the quantum correlations arise via Born's rule: $p(ab|xy)=\tr{(M_{a|x}\otimes M_{b|y})\rho_{AB}}{}$. Here, $M_{a|x}$ (resp.~$M_{b|y}$) are the POVM elements of Alice's (resp.~Bob's), and $\rho_{AB}$ is the density matrix representing the quantum state of the shared system.  Previous work has investigated whether nonlocality is a typical property of pure entangled states~\cite{Drumond2012,deRosier2017,Lipinska2018,Drumond2018,de2020strength}. The general conclusion is that as either the number of parties or the number of measurement settings increases, typicality of nonlocality exceeds $99\%$ for pure states and projective measurements. Moreover, it has been shown that for any pure bipartite entangled state, typicality tends to unity as the number of measurement settings tends to infinity~\cite{Lipinska2018}, while actual violations of Bell inequalities grow increasingly small with the local dimension of the quantum system~\cite{Drumond2012,Drumond2018}. 
	
	To compare the results of typicality of contextuality to typicality of nonlocality, we can use the fact~\cite{schmid18ineq,wright2023invertible} that a bipartite Bell scenario can be conceptualized as a remote state preparation followed by a local measurement~\cite{Li2013}, which resembles the structure of a prepare-and-measure scenario. In such a reformulation, the system that is studied is the subsystem in Bob's lab. The possible preparations of that system then correspond to the states $\rho_{a|x}=\frac{1}{p(a|x)}\tr{(M_{a|x}\otimes \id)\rho_{AB}}{A}$ that Alice steers on Bob's system and are therefore labeled by the settings and outcomes of the measurement $\{M_{a|x}\}$ . The measurements in the associated prepare-and-measure scenario are simply $\{M_{b|y}\}$. Naturally, this reformulation artificially introduces correlations among the preparations, as they are not independent. In particular, they now satisfy the no-signaling condition given by $\sum_a p(a|x)\rho_{a|x}=\rho_B$ with $\rho_B=\tr{\rho_{AB}}{A}$.  This fact has significant consequences for studying typicality, as we explain below. Hereon, we focus on typicality of contextuality for pure preparations and projective measurements, as this is the case usually considered for studying typicality of nonlocality.
	
	Consider a bipartite Bell scenario with two dichotomic measurements, which corresponds to a prepare-and-measure scenario with four states and two measurements. Recall that for $n=4$, as per Lemma~\ref{lemma:TipiZeroForLI}, typicality of contextuality is zero. However, in the Bell setting the reported typicality is $5.32\%$~\cite{de2020strength}.  This might look contradictory, but notice that this is just a consequence of the specifics of a Bell scenario: in a prepare-and-measure scenario that underpins a Bell scenario, the states $\rho_{a|x}$ cannot all be randomly sampled (like we generically do in typicality assessments here), since they must obey the no-signaling condition. 
	
	Therefore, one must bring in some caveats when trying to compare the two situations. Two options stand out as possible avenues to pursue this comparison: on the one hand, one might equip the sampling techniques in a prepare-and-measure scenario with extra constraints to effectively implement a sampling method for a Bell scenario; on the other hand, one might compare the two different scenarios and see how the peculiarities of one may open the door to observing things that the other doesn't display. In this paper we will explore the latter. 
	
	\begin{resu}
		Typicality results in the prepare-and-measure scenario cannot be quantitatively translated into typicality statements for Bell scenarios (or vice versa) if the extra linear constraints between the states are not taken into account.
	\end{resu}

	A first thing to observe is that, qualitatively, both Bell and prepare-and-measure scenarios behave similarly: typicality increases with the number of states and effects, as well as with the number of measurement settings. However, the main difference between the two scenarios is that the typicality of contextuality reaches $99\%$ significantly faster. In Ref.~\cite[Table~I]{de2020strength}, authors report typicality of two-qubit systems for a number of measurement settings ranging from 2 to 11 per party, with the maximum typicality of nonlocality being $76.80\%$ for 11 settings. Most of the corresponding prepare-and-measure scenarios (i.e., prepare-and-measure scenarios with corresponding values for $n$ and $m$, but randomly sampled states and effects), however, have typicality above $99\%$. For example, the Bell scenario with 5 measurement settings per party (with reported typicality $50.04\%$) corresponds to a prepare-and-measure scenario with $(n=10,m=5)$, which we show in Section~\ref{sec:results1} (see Fig.~\ref{fig:3dplots}(a)) to have typicality $99\%$ (for pure states and projective measurements, like in the Bell case).  
	
	\begin{nresu}
		Typicality of contextuality reaches $t_{(n,m,d;N)} \approx 100\%$ faster than the typicality of nonlocality. 
	\end{nresu}

	We therefore conclude that the qualitative behavior of typicality is similar for contextuality and nonlocality tests (with the exception of the $n=4$ case). However, witnessing nonclassicality with high typicality ($t_{(n,m,d;N)}\approx 100\%$) is easier through contextuality experiments than through nonlocality ones (at least for our sampling methods).  This suggests that, for performing nonclassical experiments with comparable resources, it is advantageous to probe contextual features rather than nonlocal ones: besides not requiring space-like separation or the preparation of bipartite entangled states, tests of contextuality can use fewer (or less stringently chosen) procedures to generate nonclassicality.

	\section{Application to parity-oblivious multiplexing} \label{sec:POM}
	
	All of our results from the previous sections show that contextuality is fairly common in qubit experiments involving random states and measurements. We now explore how this translates into typicality and strength of \emph{advantage} for a particular quantum information processing task known as parity-oblivious multiplexing (POM) task~\cite{spekkens09}. 
	
	It is known that contextuality powers the quantum advantage in POM, and it was recently shown that the robustness of contextuality estimated by the linear program employed in this paper is closely connected to the success rate of this task \cite{fonseca25}. Consider a POM task with $2^k$ qubit states and $k$ binary measurements. Let $s_{NC}$ denote the optimal noncontextual success rate for this task, and let $s$ be the maximal success rate obtained with a given strategy. It was shown in Ref.~\cite{fonseca25} that any contextual strategy, i.e., any strategy with a non-zero robustness of contextuality to depolarization $(r>0)$, provides the advantage over the optimal classical strategy in this task, with the success rate being connected to the robustness as follows: 
	
	\begin{equation}\label{eq:robustness-success}
		s=\frac{\frac12 r-s_{NC}}{r-1}.
	\end{equation}
	
	This correspondence implies that if a particular strategy typically exhibits contextuality, it also typically exhibits quantum advantage in the POM task. In this section, we examine various strategies for POM, and we assess both the typicality of finding a quantum advantage and how strong this advantage is.
	
	We focus on a POM task where one aims to encode 3 classical bits on a qubit (i.e., $k=3$). Then, the encodings that maximize the success rate are precisely the ones that optimize the volume of the set of states, i.e., the ones that form a cube on the Bloch sphere:
	\begin{equation}\label{eq:POM}
		\rho_{(x_1,x_2,x_3)}=\frac12\begin{bmatrix}
			1+\frac{(-1)^{x_3}}{2} & \sqrt{\frac32}\frac{(-1)^{x_1}-i(-1)^{x_2}}{2}\\
			\sqrt{\frac32}\frac{(-1)^{x_1}+i(-1)^{x_2}}{2} & 1-\frac{(-1)^{x_3}}{2}
		\end{bmatrix},
	\end{equation}
	as depicted in Fig.~\ref{fig:POMstates}, where $x_1,x_2,x_3\in\{0,1\}$. Notice that these states satisfy the parity-obliviousness condition in a POM game. This condition is related to what effective mixed state the system is prepared in when we consider an ensemble of states from $\{\rho_{(x_1,x_2,x_3)}\}$. For instance, one could consider the ensemble $\{\rho_{(x_1,x_2,x_3)}\}_{x_1\oplus x_2\oplus x_3=0}$ defined by the parity condition $x_1\oplus x_2\oplus x_3=0$ of the string $(x_1,x_2,x_3)$. The idea of parity-obliviousness is that, when only a parity of string $(x_1,x_2,x_3)$ is known, then we demand that the effective mixed state of the ensemble remains the same for the possible values of the parity—i.e., by knowing the parity of the string, we can't gain information on what ensemble has been prepared. 
	
	\begin{figure}[htb!]
		\centering
		\includegraphics[width=\linewidth]{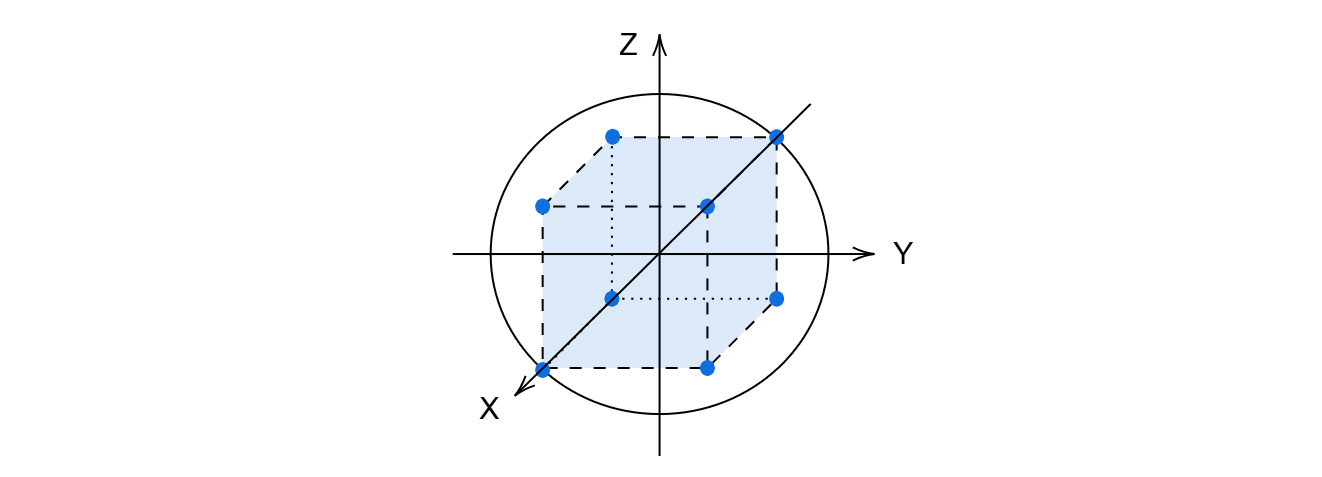}
		\caption{Optimal states for the 3-to-1 parity-oblivious multiplexing task. The optimal measurements for the task lie on the $X$, $Y$, and $Z$ axes.}
		\label{fig:POMstates}
	\end{figure}
	
	The measurements that maximize the success rate for decoding the message are those in the $\hat{X}$, $\hat{Y}$, and $\hat{Z}$ Pauli bases. It is known that the optimal success rate for a noncontextual implementation of this task reaches $s_{NC} = \frac12\left(1+\frac{1}{k}\right)\approx0.667$ for $k=3$, while the optimal quantum implementation we just described yields a success rate $s_Q =\frac12\left(1+\frac{1}{\sqrt{k}}\right)\approx 0.788$, yielding an advantage $s_Q/s_{NC}-1$ of about 18.3\%. From Eq.~\eqref{eq:robustness-success}, one can calculate that the optimal quantum advantage corresponds to the robustness $r = 0.42$.
	
	With this information in mind, let us start our numerical study. We consider three different cases, wherein in all cases we leave the states fixed as the optimal states described above, but wherein we have three different methods for sampling measurements (discussed below). First, we check the typicality of contextuality for each of the three cases. Then, we compute the average success rate ($\bar{s}=\frac{1}{N}\sum_is_i$), the average advantage ($\bar{s}/s_{NC}-1$), and the average robustness of contextuality ($\bar{r} = \frac{1}{N}\sum_ir_i$) for each of the considered scenarios. 
	We denote the standard deviation of the average advantage and the standard deviation of the average robustness of contextuality as $\sigma_s$ and $\sigma_r$, respectively.

	\
	
	\noindent \textbf{Case i)} 
	Instead of performing the optimal measurements, the decoding party performs three completely random projective measurements. This could happen if the decoding party has no control over which measurements are implemented. With this strategy, one has a $98.6\%$ chance of obtaining contextuality, which implies that this strategy has a $98.6\%$ chance of obtaining a quantum advantage. We can then use Eq.~\eqref{eq:robustness-success} to quantify the strength of the advantage in each case, and from this compute the average advantage.

	\begin{nresu} 
		Consider the prepare-and-measure scenario $(n=8,m=3,d=2;N=10^6)$ with $m=3$ random projective measurements, but taking the $n=8$ states to be those optimal for POM. The typicality of contextuality for this scenario is $t_{(n,m,d;N)} \approx 98.6\%$. The average success rate in the POM task is $\bar{s}=0.72$ 
		($\sigma_s=0.02$)
		, which is about $8\%$ above the optimal noncontextual strategy. The average robustness of contextuality is $\bar{r} = 0.22$ ($\sigma_r=0.07$).
	\end{nresu}
	
	This result demonstrates two features of the quantum advantage in the POM task. First, high typicality of contextuality does not imply a strong quantum advantage. This is unsurprising, as typicality of contextuality means that the average success rate is typically above the noncontextual one but does not determine the strength of this advantage. Second, it is interesting to note that although the improvement over the optimal noncontextual strategy is much smaller than in the optimal strategy case (which yields an $18.3\%$ improvement over the noncontextual strategy), one still finds a significant advantage on average, even when the decoding party performs random projective measurements.
	
	\
	
	\noindent \textbf{Case ii)} 
	The decoding party does not implement random projective measurements but rather random POVMs. This may happen when the decoding party also lacks any control over the sharpness of the measurements or when the channel between the encoding and the decoding parties is subject to noise. The typicality 
	of contextuality in this case is equal to $15.2\%$, which is significantly lower than for case i) with random projective measurements. The results of the analysis of the average success rate are summarized below.
	
	\begin{nresu} 
		Consider the prepare-and-measure scenario $(n=8,m=3,d=2;N=10^6)$, with $m=3$ randomly sampled dichotomic POVMs but the $n=8$ fixed states used in the optimal strategy for POM. The typicality of contextuality for this scenario is $t_{(n,m,d;N)} \approx 15.2\%$. The average success rate in the POM task is $\bar{s}=0.669$ ($\sigma_s=0.006$), which is about $0.3\%$ above the optimal noncontextual strategy. The average robustness of contextuality is $\bar{r} = 0.01$ ($\sigma_r=0.03$).
	\end{nresu}
	
	\
	
	\begin{table*}[htb]
		\centering
		\begin{tabular}{c|c|c|c}\hline
			Strategy & Typicality ($t$) & Average robustness ($\bar{r}$) & Average advantage ($\bar{s}/s_{NC}-1$)  \\\hline
			Optimal & $-$ & $0.42$ & $18.3\%$  \\
			Lack of RF &  $\approx$ 100\% & 0.3  & 11\% \\
			Random PVMs & 98.6\% & 0.22 & 8\% \\
			Random POVMs & 15.2\% & 0.01 & 0.3\% \\\hline
		\end{tabular}
		\caption{Typicality of contextuality, average advantage over the optimal classical strategy, and average robustness of contextuality to depolarizing noise of 3-to-1 parity-oblivious multiplexing tasks in a qubit with optimal encoding strategy but suboptimal measurements. The first row includes the data in the optimal quantum scenario, for reference. Average advantages had a variance of about $\sigma_s=0.02$ for the scenarios with rotated optimal measurements and with random projective measurements, while $\sigma_s=0.006$ for the scenario with random POVMs. For the average robustness, the variance was of about $\sigma_r=0.07$ for the former cases, and $\sigma_r=0.03$ for the latter.}
		\label{tab:POM}
	\end{table*}
	
	\noindent \textbf{Case iii)} There is an arguably more realistic situation that we could also consider, where the obstacle for the encoding and decoding parties to succeed at POM is not due to their local experimental limitations, but rather due to a lack of common reference frame (RF) between them. In this case, we take the preparations and the measurements to be the optimal pure states and projective measurements spanning a set of mutually unbiased bases. The lack of a common reference frame is then captured by rotating these optimal measurements with respect to the same (fixed but arbitrary) axis and angle (to implement the rotations, we sample Haar-random $SU(2)$ unitaries; see Appendix~\ref{app:qutip} for details). The typicality of contextuality is almost $100\%$ in this case: 
	
	\begin{nresu}
		Consider a POM task where the encoder prepares the optimal states and the decoder performs the optimal measurements, but where they lack a common reference frame. The typicality of contextuality in this $(n=8,m=3,d=2;N=10^6)$ prepare-and-measure scenario (for randomly chosen reference-frame misalignments) is $t_{(n,m,d;N)}\approx 100\%$.  The average success rate is $\bar{s}=0.74$ ($\sigma_s=0.02$), with an average robustness of $\bar{r} = 0.3$ ($\sigma_r=0.07$). The average success rate is $11\%$ above the classical optimal rate.
	\end{nresu}
	
	This result again shows that large typicality of contextuality does not imply large quantum advantage. Nonetheless, this is a surprisingly high advantage given the setup, implying that as long as the decoding agent can ensure they are locally performing measurements with the features of the optimal implementation (the angles between the measurements and the sharpness), the lack of a shared reference frame will not jeopardize the quantum advantage considerably. 
	
	All the numerical results presented in this section are summarized in Table~\ref{tab:POM}.

	\section{Conclusions}
	
	In this work, we begin to investigate the question of how typical contextuality is in prepare-and-measure scenarios. In this setting, we find that contextuality is quite typical even for relatively small numbers of states and measurements, and we show how these numbers must increase as noise increases.  This novel investigation is only possible thanks to efficient numerical tools introduced in recent years~\cite{selby2024linear}. Our study, moreover, goes beyond what has been done for studying typicality of nonlocality in that we probe scenarios with mixed states and unsharp measurements. 
	
	Previous studies of typicality of nonclassicality have focused primarily on Bell nonlocality. By noting that a Bell scenario can be viewed as remote state preparation followed by a measurement, we were able to compare our results, with the caveat that the no-signaling constraint introduces nontrivial linear dependencies. In nearly all cases, the typicality of contextuality reaches significantly higher values than that of Bell nonlocality, indicating that generating nonclassicality is much more common in contextuality scenarios. Moreover, we observe that the typicality of contextuality can approximate its maximal value for a smaller number of random measurements than the one needed in Bell scenarios. In this sense, contextuality is a more experimentally accessible form of nonclassicality than Bell nonlocality.
	
	The fact that contextuality appears with high typicality in simple random experiments could lead to the mistaken idea that meaningful quantum advantages are easily attained even in careless implementations. We illustrated why this is not the case by analyzing one particularly important instance of a contextuality-powered task, demonstrating that randomly chosen projective measurements in a parity-oblivious multiplexing experiment display fairly high typicality of contextuality, even when the average success rate is low compared to the optimal case. This shows indeed that high typicality of contextuality does not translate into typically high quantum advantage. By adding additional constraints to how we sample the measurements in POM tasks, we can consider how different obstacles to experimental implementation may affect the ability to obtain nonclassical experiments as well as how these affect the ability to obtain advantage. We conclude that reference frame misalignment between the encoding and decoding parties still permits an advantage in this task, although the advantage is significantly reduced compared to the optimal strategy. 
	
	Our results also provide insight into how one can generate contextuality in experimental setups. Although pure states and sharp measurements are never truly achievable in an experiment, this lack of purity can easily be compensated by increasing the number of sampled states and performed measurements. We provide a visualization of this trade-off, along with a numerical tool that predicts the minimal number of random preparations required for typicality of contextuality of at least $99\%$ given:  (i) the number of random binary measurements, (ii) the experimental upper and lower bounds on the purity of the states and/or POVMs that can be produced, and (iii) the number of repetitions $N$ that this experiment is performed~\cite{github}.
	
	Our results also suggest a useful strategy for engineering contextuality: one should aim to generate as many operational equivalences as possible; for example, ensuring that an experiment spans only a subspace of the state/effect space can significantly enhance the odds of observing contextuality (so, for instance, instead of trying to spread out the states across the full 3-dimensional Hilbert space of a qutrit, one should try to restrict the sampled states to a qubit subspace in order to witness contextuality more easily). The repository we provide serves as a versatile toolbox for future theoretical and experimental research. It can guide analytical investigations into the minimal number of measurements required to generate contextuality for a given set of preparations and can be extended to explore the typicality of contextuality in higher-dimensional systems, though this requires greater computational power. For laboratory tests of contextuality, the toolbox allows users to incorporate laboratory-specific purities, enabling studies analogous to ours to be tailored to particular setups and directly compared with experimental data. Moreover, by restricting how procedures are sampled (as we did for MUB measurements in POM tasks), one can see how different designs or experimental constraints affect the typicality of witnessing nonclassicality.
	
	Some new avenues open up with this kind of investigation. For instance, we discussed briefly in Section~\ref{sec:results1} how our methods can be applied to the study of nonclassicality of individual processes, as initiated in Refs.~\cite{zhang2025reassessing,zhang2025quantifiers}. Considerably more work could be done to study the typicality of nonclassicality of single processes by continuing in this direction. Moreover, it would be interesting to see how typicality of contextuality behaves in more general frameworks (such as the one of generalized probabilistic theories) and how these numerical tools may be employed to investigate novel notions in the contextuality program~\cite{shadows}.
	
	\section*{Acknowledgments}
	
	We thank Pawe\l~Cie\'sli\'nski for the discussions on typicality of Bell nonlocality and Piotr Mironowicz for the support with the cluster. V.P.R., B.Z., and R.D.B acknowledge support by the Digital Horizon Europe project FoQaCiA, Foundations of quantum computational advantage, GA No.~101070558, funded by the European Union, NSERC (Canada), and UKRI (UK). 
	D.S. and R.D.B. were supported by Perimeter Institute for Theoretical Physics. Research at Perimeter Institute is supported in part by the Government of Canada through the Department of Innovation, Science and Economic Development and by the Province of Ontario through the Ministry of Colleges and Universities. 
	JHS was funded by the European Commission by the QuantERA project ResourceQ under the grant agreement UMO2023/05/Y/ST2/00143.
	This work is partially carried out under IRA Programme, project no.~FENG.02.01-IP.05-0006/23, financed by the FENG program 2021-2027, Priority FENG.02, Measure FENG.02.01., with the support of the FNP. We acknowledge the use of a computational server financed by the Foundation for Polish Science (IRAP project, ICTQT, contract no. 2018/MAB/5/AS-1, co-financed by EU within Smart Growth Operational Programme). Some figures were prepared using Mathcha.

\bibliographystyle{apsrev4-2}
\bibliography{main}
	\appendix

	\section{Simplex embedding linear program}\label{app:LP}
	
	The linear program provided in Ref.~\cite{selby2024linear} relies on the equivalence between the existence of a noncontextual ontological model for an operational scenario and the existence of a simplex embedding for its associated GPT fragment~\cite{schmid2021characterisation,Schmid2024structuretheorem}. In this section we summarize the results of Ref.~\cite{selby2024linear} and comment on how they are implemented in Ref.~\cite{cavalcanti-git}. We assume some familiarity of the reader with the basics of generalized probabilistic theories and refer to Ref.~\cite{janotta2014generalized} if additional background is needed. \\
	
	Given a set of all possible preparations $\mathcal{P}$ and measurements $\mathcal{M}$ comprising a prepare-and-measure scenario, it is possible to construct the associated GPT $(\Omega,\mathcal{E},V)$, in which $V$ is a real vector space and $\Omega$ a convex set containing the representation of the states in this space, in such a way that the operational equivalences between preparations are captured by convex combinations of vectors in $\Omega$. The same applies to $\mathcal{E}$: it is a convex set in the dual $V^*$, such that operational equivalences between measurement outcomes are captured by convex combinations in $\mathcal{E}$ yielding the same vectors. These sets are expected to satisfy some properties, such as containing normalized counterparts for any state, complementary counterparts for any effects, a unitary effect, and the feature that $\Omega$ and $\mathcal{E}$ are separating for each other. The statistics of the operational scenario are recovered by acting effects $e\in\mathcal{E}$ on states $s\in\Omega$, such that
	\begin{equation}
		p(k|M,P)= e_{k|M}[s_P].
	\end{equation}
	
	Notice that a prepare-and-measure scenario often does not constitute a full theory: preparations or measurement outcomes might be present without their normalized counterparts, and almost always the sets $\Omega$ and $\mathcal{E}$ won't be tomographically-complete for each other. We therefore refer to $(\Omega,\mathcal{E},V)$ as a \emph{GPT fragment} rather than a GPT. A natural implication is that the sets $\Omega$ and $\mathcal{E}$ do not span the full spaces $V$ and $V^*$, but rather subsets $S_\Omega:=\mathsf{Span}(\Omega)$ and $S_\mathcal{E}:=\mathsf{Span}(\mathcal{E})$. Most often, $S_\mathcal{E}\neq S_\Omega^*$. One could therefore naturally describe $\Omega$ and $\mathcal{E}$ in their respective subspaces, i.e., define projectors $P_\Omega:V\to S_\Omega$ and $P_\mathcal{E}:V^*\to S_\mathcal{E}$ such that $\Omega_\mathcal{A}:=P_\Omega(\Omega)$ and $\mathcal{E}_\mathcal{A}:=P_\mathcal{E}(\mathcal{E})$ are the new sets of preparations and measurements. Evidently, the effects in $\mathcal{E}_\mathcal{A}$ cannot directly act on the states in $\Omega_\mathcal{A}$, so to recover probabilities we must undo these projections. We therefore define inclusion maps $I_\Omega:S_\Omega\to V$ and $I_\mathcal{E}:S_\mathcal{E}\to V^*$, such that
	\begin{equation}
		p(k|M,P)= I_\mathcal{E}[e_{k|M}^{\mathcal{A}}](I_\Omega[s_P^\mathcal{A}]).
	\end{equation}
	
	The tuple $(\Omega_\mathcal{A},\mathcal{E}_\mathcal{A},I_\Omega,I_\mathcal{E})$ carries the exact same information as $(\Omega,\mathcal{E},V)$, and we refer to it as the \emph{accessible GPT fragment}. Working with accessible fragments is computationally convenient since one now works with sets of states and effects in smaller spaces than the full GPT. 
	
	Simplex embeddability is a property of the positive cones of states and effects, not of the sets themselves~\cite{selby23fragments}. In order to characterize these cones, one can find the set of all inequalities a vector must satisfy to live inside the corresponding cone\footnote{For polytopic GPT fragments, this set is always finite.}, that is, the set of vectors $\{h_\Omega^i\}_{i=1}^n$ in $V^*$ such that
	\begin{equation}
		h_i^\Omega(v)\geq0,\,\,\forall i\iff v\in\mathsf{Cone}(\Omega_\mathcal{A}),
	\end{equation}
	and respectively the set of vectors $\{h_\mathcal{E}^j\}_{j=1}^m$ in $V$ such that
	\begin{equation}
		w(h_\mathcal{E}^j)\geq0,\,\,\forall j\iff w\in\mathsf{Cone}(\mathcal{E}).
	\end{equation}
	Algebraically, one can concatenate these inequalities into matrices such that $H_\Omega$ contains $h_\Omega^i$ as rows. One can then express
	\begin{equation}
		H_\Omega(v)\geq_e0\iff v\in\mathsf{Cone}(\Omega_\mathcal{A}).
	\end{equation}
	Similarly, one can define the matrix $H_\mathcal{E}$ containing all the inequalities for the cone of $\mathcal{E}_\mathcal{A}$ and express an analogue equation to characterize membership in $\mathsf{Cone}(\mathcal{E}_\mathcal{A})$.  Ref.~\cite{selby2024linear} proves that assessing the simplex embeddability for the accessible fragment $(\Omega_\mathcal{A},\mathcal{E}_\mathcal{A},I_\Omega,I_\mathcal{E})$ is an instance of a linear program. Given $(H_\Omega,H_\mathcal{E},I_\Omega,I_\mathcal{E},D)$, where $D$ is a complete depolarizing channel, one can assess simplex embeddability by implementing the following:
	\begin{align}
		\text{minimize}\qquad & r\\
		\text{such that} \qquad & (1-r)I_\mathcal{E}^T\cdot I_\Omega+rI_\mathcal{E}^T\cdot D\cdot I_\Omega=H_\mathcal{E}^T\cdot\sigma\cdot H_\Omega,\\
		& \sigma\geq_e0\text{ is an }m\times n\text{ matrix}.
	\end{align}
	As mentioned in the main text, $r$ is a qualitative certifier of contextuality, since $r>0$ implies the impossibility of a noncontextual ontological model for the original prepare-and-measure scenario.\\
	
	\begin{table*}[htb!]
		\begin{tabular}{c|c|c|c}\hline
			Step & Input & Output & Execution \\\hline
			0 & $(\mathcal{P},\mathcal{M})$ & $(\Omega,\mathcal{E},\mu,u)$ & \parbox[t]{0.6\linewidth}{Construct the GPT fragment by decomposing states and effects into their Gell-Mann coefficients. Additionally, provide the maximally mixed state $\mu$ and the unit effect $u$. This step can be skipped if the input is already a GPT fragment.}\\\hline
			1 & $(\Omega,\mathcal{E})$ & $(\Omega_\mathcal{A},\mathcal{E}_\mathcal{A},I_\Omega,I_\mathcal{E})$ & \parbox[t]{0.6\linewidth}{Assess the subspaces spanned by $\Omega$ and $\mathcal{E}$ via reduced row echelon form. Construct the projectors $(P_\Omega,P_\mathcal{E})$ and the inclusions $(I_\Omega,I_\mathcal{E})$ as their pseudoinverses.}\\\hline
			2 & $(\Omega_\mathcal{A},\mathcal{E}_\mathcal{A})$ & $(H_\Omega,H_\mathcal{E})$ & \parbox[t]{0.6\linewidth}{Characterize cone facets via Motzkin's double description method using \texttt{cdd}.}\\\hline
			3 & $(H_\Omega,H_\mathcal{E},I_\Omega,I_\mathcal{E},\mu,u)$ & $r$ & \parbox[t]{0.6\linewidth}{Construct the depolarizing channel $D$ from $\mu$ and $u$. Perform the linear program from Ref.~\cite{selby2024linear} with \texttt{cvxpy}.}\\\hline
			
		\end{tabular}
		\caption{Summary of the steps and numerical techniques employed in implementation~\cite{cavalcanti-git}.}
		\label{table:LP}
	\end{table*}
	
	The implementations available at Refs.~\cite{elie-git,cavalcanti-git} are both equipped to perform every step of the above optimization program. Since we employ the implementation in Ref.~\cite{cavalcanti-git} for this project, we here focus on its technical details, which are summarized in Table~\ref{table:LP}. Notice that the original code also outputs the ontological model for the partially depolarized scenario, i.e., the set of ontic states $\{\mu_P(\lambda)\}_P$ and response functions $\{\xi_{k|M}(\lambda)\}_{k,M}$ that provide the noncontextual explanation for the partially depolarized scenario. Since in this work we are only interested in the certification of contextuality (i.e., the value of $r$), we skip this step entirely.
	
	There are multiple parameters that can impact the assessment of contextuality provided by this program. Any library for polytope conversion, such as \texttt{cdd}, is sensitive to small fluctuations in the entries of the analyzed matrices, as well as large matrices with multiple repeated rows or columns~\cite{fukuda1996, cddman}. Moreover, optimization problems are also sensitive to the convergence criteria of the employed library and the multiple solvers available within the same library. Some of these aspects are investigated and addressed in Appendix~\ref{app:threshold}.
	
	\section{Sampling states, effects, and unitaries}\label{app:qutip}
	
	To generate random states and effects, we rely on standard methods from quantum information theory. In this appendix we explain the basic mathematical constructions that underlie the implementations in the QuTiP toolbox~\cite{lambert2024qutip5quantumtoolbox}.
	
	\subsection{Random pure states and unitaries}
	Let $|\psi\rangle \in \mathbb{C}^d$ be a pure state in a Hilbert space of dimension $d$. The uniform distribution of pure states is defined by the Haar measure on the unitary group $U(d)$, which is the unique probability measure that is invariant under both left and right multiplication by unitaries. A random Haar-distributed state can be generated as
	\begin{equation}
		|\psi\rangle = U |0\rangle,
	\end{equation}
	where $U \in U(d)$ is Haar-distributed and $|0\rangle$ is the computational basis state. This ensures that the ensemble of states $|\psi\rangle$ is invariant under arbitrary unitary transformations and thus samples uniformly over the unit sphere $S^{2d-1}$ in Hilbert space.
	
	To generate a Haar-random unitary matrix $U \in U(d)$, first construct a $d \times d$ complex matrix $Z$ with independent standard normal entries. Next, perform a QR decomposition $Z = Q R$, and rescale the diagonal of $R$ to have unit modulus by defining $\Lambda = \mathrm{diag}(R)/|\mathrm{diag}(R)|$. The Haar-random unitary is then given by $U = Q \Lambda$ ~\cite{mezzadri06}.
	
	For the single-qubit case $d=2$, generating a random unitary simplifies, since any unitary can be parametrized by three Euler angles as
	\begin{equation}
		U(\alpha,\beta,\gamma) = U_z(\alpha)\, U_y(\beta)\, U_z(\gamma),
	\end{equation}
	with $U_i(\theta) = e^{-i \theta \sigma_i/2}$ and $\sigma_i$ being Pauli matrices. The corresponding Haar measure is
	\begin{equation}
		d\mu(U) = \sin(\beta)\, d\alpha\, d\beta\, d\gamma,
	\end{equation}
	which specifies how the angles must be sampled with the correct weight in order to obtain Haar-random unitaries~\cite{Zyczkowski1994}. Projective measurements for a system of dimension $d$ can thus be sampled by uniformly sampling unitaries in $U(d)$ and applying them to a orthonormal basis in this Hilbert space. For the qubit, this is equivalent to sampling a single random pure state $\ket{\psi}$ and constructing the measurement $\{\ket{\psi}\!\!\bra{\psi},\mathbbm{1}-\ket{\psi}\!\!\bra{\psi}\}$.
	
	\subsection{Random mixed states and POVMs}
	
	A mixed state in a $d$-dimensional Hilbert space is described by a density operator $\rho \in \mathcal{D}(\mathbb{C}^d)$. To generate random mixed states, one must choose a probability measure over all $d \times d$ density operators, and since no unique choice exists, different ensembles of random density matrices are studied. Here, we use the method of Ref.~\cite{Bruzda2009}. We first construct $G$, a $d \times d$ matrix belonging to the Ginibre ensemble~\cite{Ginibre1965}, i.e., a matrix whose entries are independent complex Gaussian variables. By forming
	\begin{equation}
		\rho = \frac{G G^\dagger}{\mathrm{Tr}(G G^\dagger)},
	\end{equation}
	one obtains a valid full-rank density matrix. Matrices generated in this way form the Hilbert-Schmidt ensemble of random density operators. In our work we restrict to full-rank states, since lower-rank density matrices form a set of measure zero~\cite{Bengtsson2017}.
	
	In order to sample mixed states constrained within a range of purities, we employ a rejection method. That is, given an upper bound $u$ and a lower bound $l$ to the purity of the mixed state, we sample a random full-rank density matrix $\rho$ employing the previously described method and then check whether $l\leq\Tr{\rho^2}\leq u$. If this doesn't hold, we simply sample another $\rho$ until the condition is met. Although this is not a uniform distribution over the set of all possible density operators with suitable purity, there is no such Haar measure established for this kind of set. 
	
	A different sampling method for states within some purity range could consist of a Markov Chain Monte Carlo algorithm constrained to fixed purities~\cite{MCMC}. Such a method might be more efficient if the interval $[l,u]$ is quite narrow, but it will still not ensure a uniform distribution within the specific set.  Since we do not consider narrow sets in our computations, we regard the rejection method as sufficient.
	
	In order to sample dichotomic POVMs, we employ the method defined in Ref.~\cite{kukulski2021}. That is, we sample pairs $(G_1,G_2)$ of $d\times d$ complex matrices belonging to the Ginibre ensamble, and then proceed to construct the positive semidefinite elements
	\begin{equation}
		E_1=G_1G_1^\dagger,\quad E_2=G_2G_2^\dagger,\quad S=E_1+E_2.
	\end{equation}
	The POVM is therefore given by the normalized elements
	\begin{equation}
		M=\{M_1:=S^{-1/2}E_1S^{-1/2},\ M_2:=S^{-1/2}E_2S^{-1/2}\},
	\end{equation}
	which ensures that the elements sum up to the identity. As demonstrated in Ref.~\cite{kukulski2021}, this measure can be generalised to quantum channels and is equivalent to other commonly considered measures. The implementation we provide in Ref.~\cite{github} also admits sampling POVMs with a larger number of outcomes.
	
	To sample POVMs with bounded sharpness, we employ the same rejection method as the one employed for density operators, with the caveat that for POVMs the quantity to be verified is given by the sharpness $\eta$~\cite{Salehi2021}, i.e., for a POVM $M:=\{M_i\}_{i=1}^k$ in a $d$-dimensional quantum system, \begin{equation}\label{eq:sharpness}
		\eta_M:=\frac1d\sum_{i=1}^k\Tr{M_i^2}.
	\end{equation}
	
	\section{Technical discussion}\label{app:threshold}
	
	The repository from Ref.~\cite{github} provides the codes employed to derive our results and relies on the linear program introduced in Ref.~\cite{selby2024linear}. Therefore, the user should make sure that the repository provided in Ref.~\cite{cavalcanti-git} is up and running before using the functions provided in Ref.~\cite{github}. \\
	
	In order to start the typicality analysis, we perform a sanity check in which 4 preparations and 2 binary measurements are randomly sampled. This corresponds to the simplest scenario in which contextuality arises \cite{spekkens09,pusey2018,catani2024}, and contextuality in this scenario can only be generated when the states and effects in the GPT fragment associated with it lie on the same hyperplane. Since this is a measure-zero subspace (as shown in Lemma~\ref{lemma:TipiZeroForLI}), our sampling should never see this for a finite number of tries $N$, and therefore the scenario $(n=4,m=2,d,N)$ should have typicality $t_{(n,m,d;N)}=0\%$ for any values of $d$ and $N$. 
	
	In practice, however, there will be scenarios in which the linear program will either obtain a suboptimal result or incorrectly converge to a value of $r$ different than 0 due to numerical fluctuations. Based on the fact that we know these cases should always yield $r=0$, we can impose a threshold on what the code should consider classical for purposes of typicality assessments. We analyze free and open-source solvers (ECOS, SCS, and CLARABEL) available through \texttt{cvxpy}, so if the user wishes to employ an alternative solver (likely a commercial one), it will be relevant to perform this analysis once again and adjust this threshold manually, or alternatively find a threshold for $r$ that can be experimentally justified.
	
	First, we analyze how the number $N$ of iterations impacts the runtime for this scenario with 4 mixed states and 2 POVMs. This analysis is crucial for practical purposes, since too large of a runtime for this simplest case will likely get even longer for the cases with multiple preparations and measurements explored in this work. We therefore plot how long the typicality assessment takes for a wide range of iteration numbers $N$ for each of the solvers. The plot can be found in Fig.~\ref{fig:runtime}. Our computations were run on an HPE Superdome Flex with 384 Intel\textsuperscript{\tiny\textregistered} Xeon Gold 6252 cores at 2.10 GHz and 12 TB of RAM, running RHEL 9.1. The typicality assessment was parallelized across 200 processes. 
	
	\begin{figure}[h]
		\centering
		\includegraphics[width=0.65\linewidth]{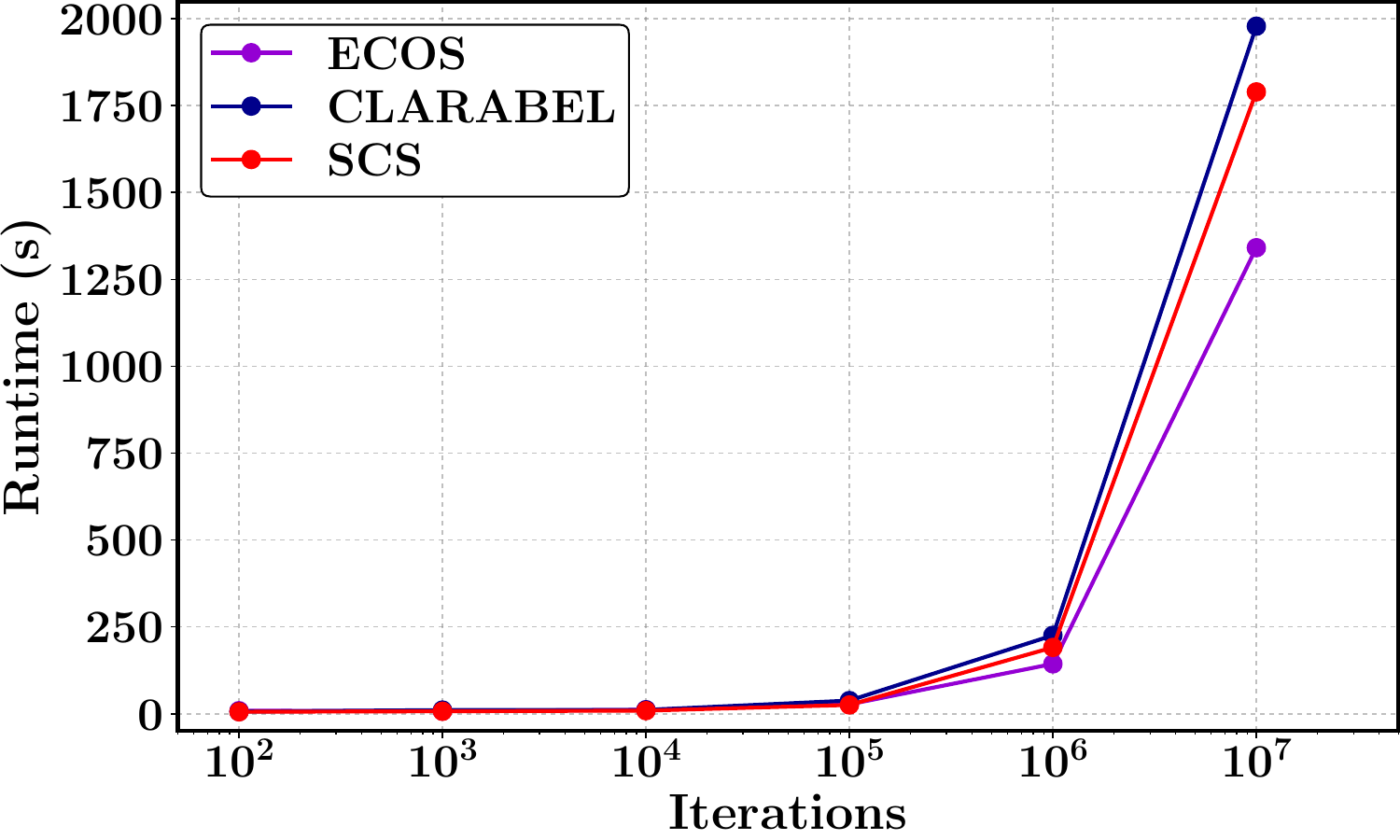}
		\caption{Runtime for typicality assessment of a scenario $(n=4,m=2,d=2,N)$ for different values of $N$ and different solvers available in the \texttt{cvxpy} library.}
		\label{fig:runtime}
	\end{figure}
	
	Any numerical simulations based on frequencies will benefit from a value of $N$ as large as possible. We can see, however, the astonishing jump in runtime from $N=10^6$ to $N=10^7$—from about 250 s for all solvers to 1250-2000 s. We can see that ECOS is the better-performing solver when it comes to runtime. Previous literature in typicality assigned $N$ between $10^9$ and $10^{10}$ as the number of iterations~\cite{de2020strength,deRosier2017}, but this is clearly not feasible in our program. We therefore constrain our attention to values of $N$ of up to $10^6$, which already represents a larger number of iterations than the average contextuality experiment~\cite{grabowecky2022,rafa-photonic,aloy24}.
	
	We then assign different non-zero thresholds to $r$ below which a scenario is deemed classical. We test it by comparing how different solvers behave under different values of the threshold for different iteration numbers $N$. For 4 pure states and 2 binary projective measurements, their performances are summarized in Fig.~\ref{fig:thresh1}. From the plots, one can see that SCS is the worst-performing solver, deeming the data from various trials as nonclassical even for the highest classical threshold considered ($r>10^{-6}$). 
	
	It is expected that SCS would perform poorly since it is best suited for linear programs with very large matrices, which is not the case for this project. ECOS and CLARABEL perform better compared to SCS, reaching the expected typicality values for the bound $r>10^{-7}$. A further tuning could be performed over the solving parameters (such as adjusting the maximum iterations for each run of the linear program or the tolerance boundaries for convergence), but for the purposes of this paper, this threshold is reasonable enough.
	
	\begin{figure*}[htb!]
		\centering
		\begin{tabular}{cc}
			a) \adjustbox{valign=t}{\includegraphics[width=0.35\linewidth]{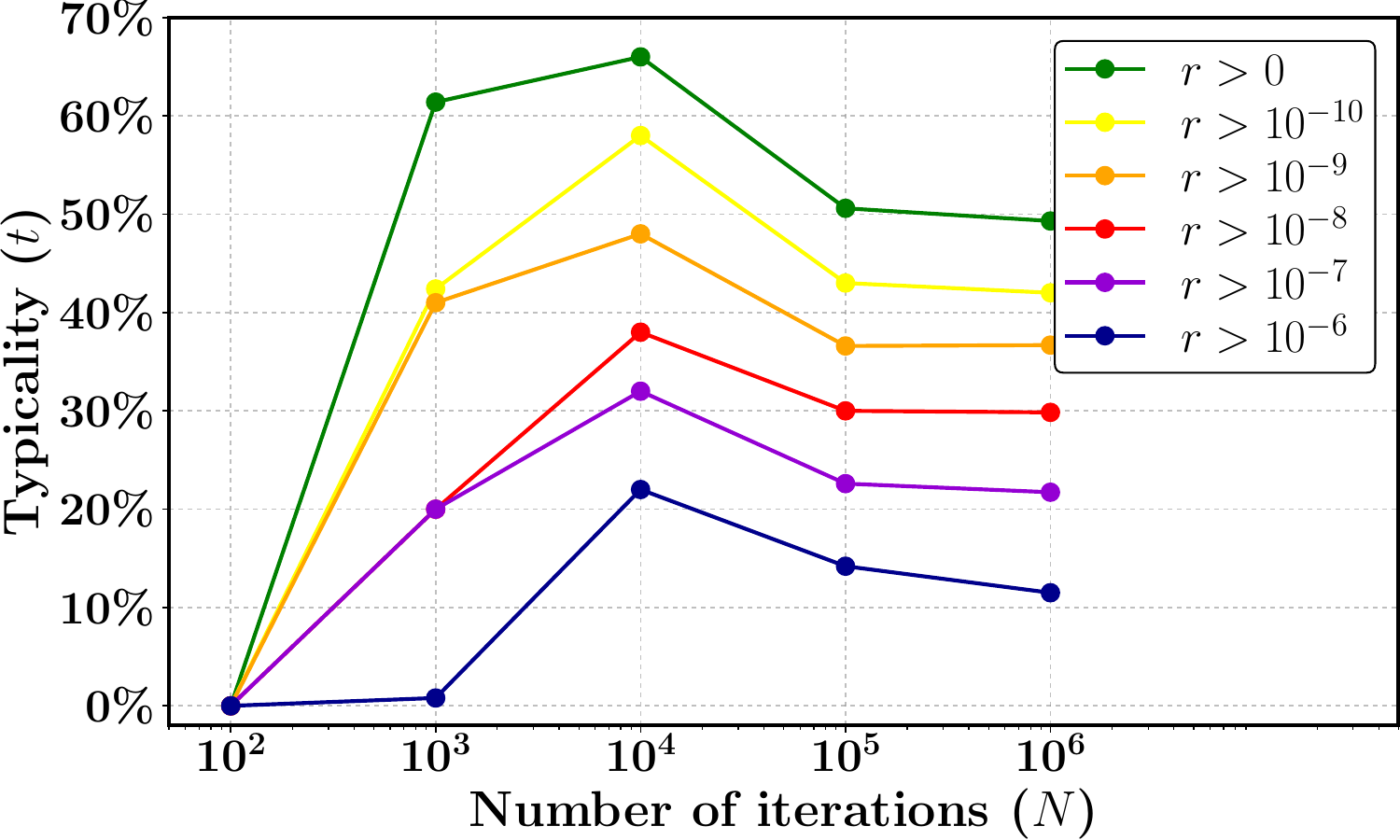}} & b) \adjustbox{valign=t}{\includegraphics[width=0.35\linewidth]{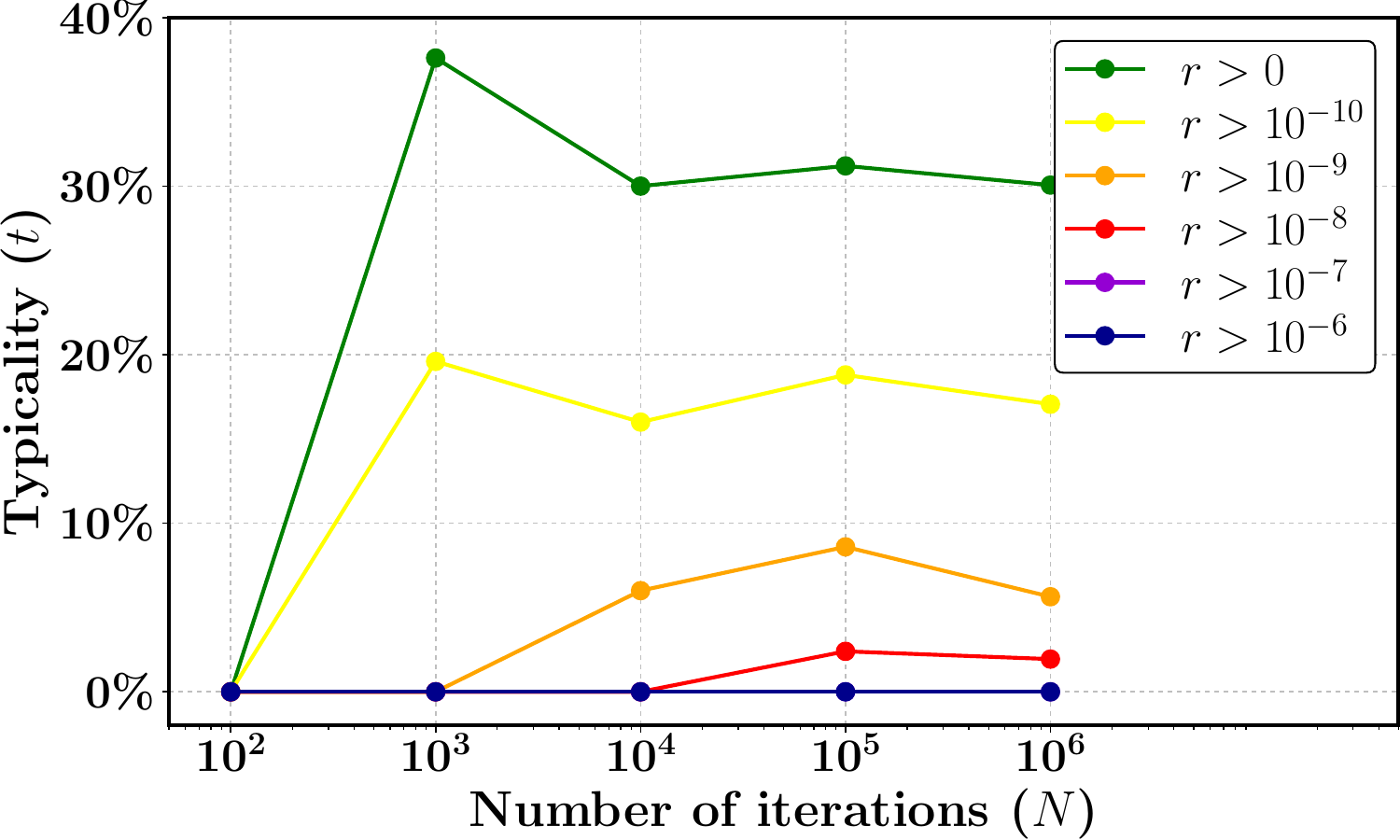}}
			\\
			c) \adjustbox{valign=t}{\includegraphics[width=0.35\linewidth]{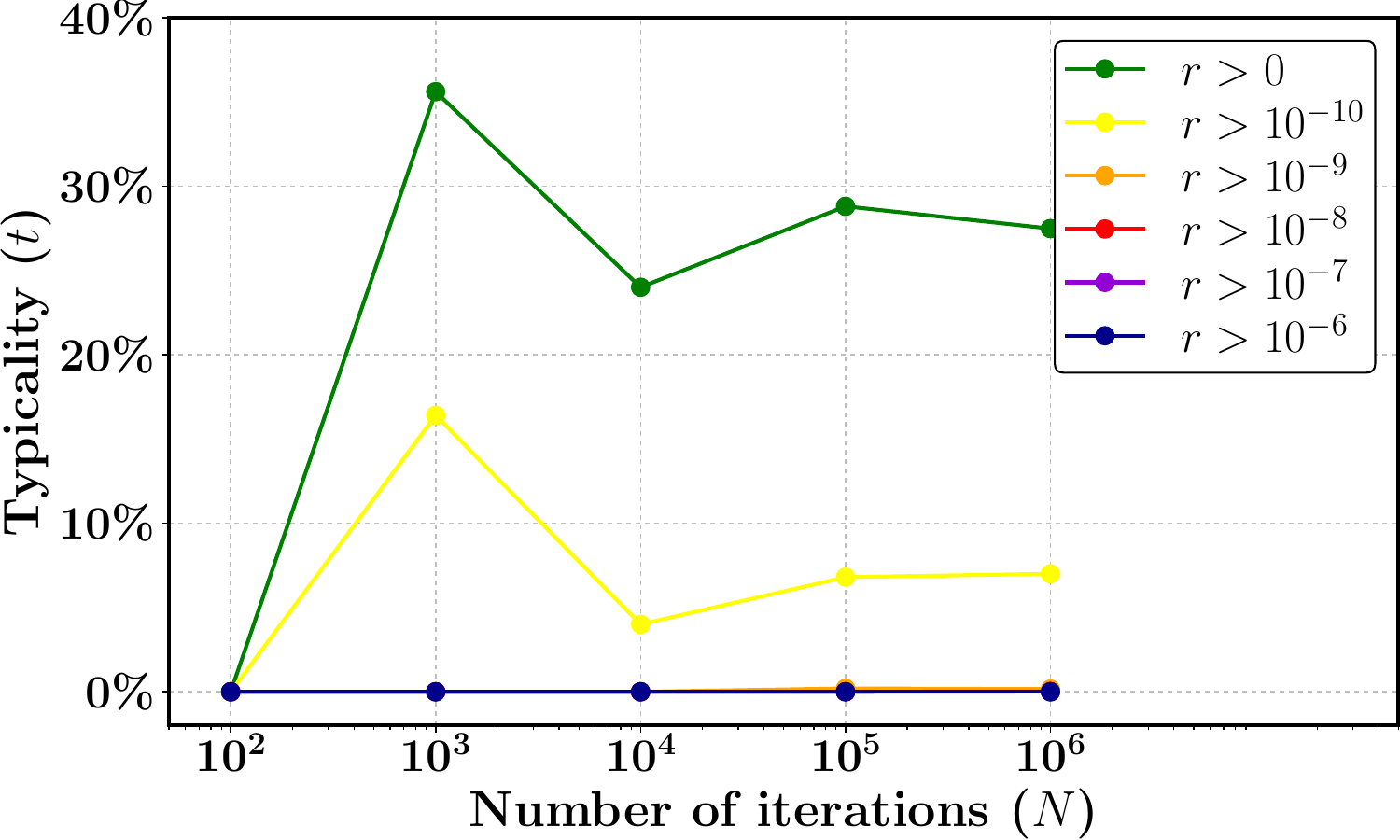}}
		\end{tabular}
		\caption{Plots of typicality of contextuality for a scenario $(n=4,m=2,d=2,N)$ with varying numbers of iterations $N$ for $n$ pure states and $m$ binary projective measurements for (a) SCS, (b) CLARABEL, and (c) ECOS, and considering different thresholds for $r$ (i.e., different options for the minimum value of $r$ over which the data is deemed contextual, which in an ideal theoretical scenario is $r=0$).}
		\label{fig:thresh1}
	\end{figure*}
	
	Since we are also concerned with scenarios with mixed states and POVMs, we also performed a comparative analysis of ECOS and CLARABEL for this case. The target value for typicality remains unchanged. The results are displayed in Fig.~\ref{fig:thresh2}, where we can see by comparing with Fig.~\ref{fig:thresh1} that both solvers reach the expected typicality for $r>10^{-7}$ both for the case of pure states and projective measurements and the case of mixed states and POVMs, although ECOS does so for lower thresholds as well. We therefore adopt this value as the threshold for classicality and settle for ECOS as the solver (due to its better performance in the runtime analysis; see Fig.~\ref{fig:runtime}), automatically calling CLARABEL instead if ECOS fails to solve the linear program for any reason. We also emphasize that this bound is of theoretical relevance, but in practical experiments $r=10^{-7}$ will seldom be detectable with the current technology. In practice, an experimentalist could (and should) manually adjust this bound to fit the capabilities of their laboratory.
	
	\begin{figure*}[htb!]
		\centering
		\begin{tabular}{cc}
			a) \adjustbox{valign=t}{\includegraphics[width=0.35\linewidth]{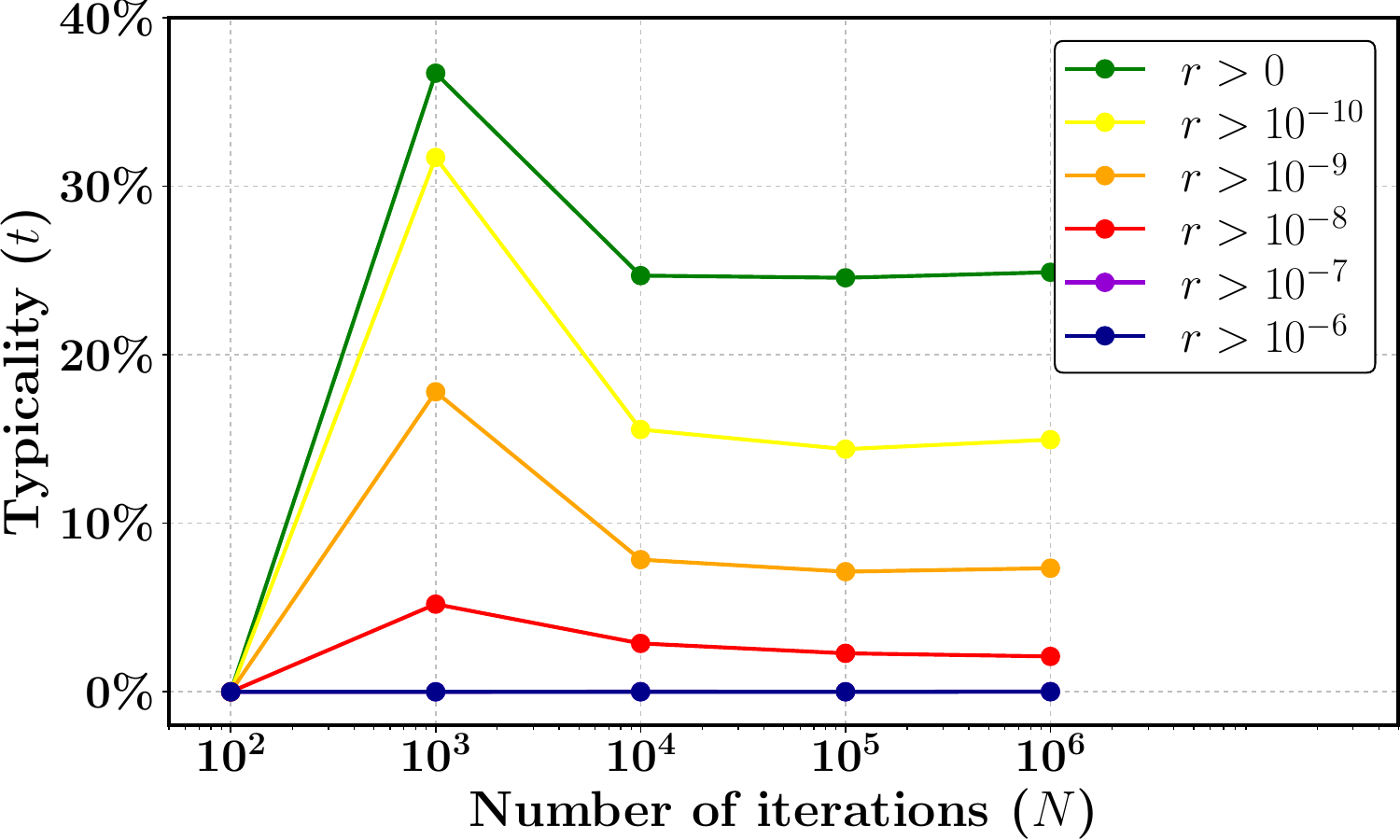}} &
			b) \adjustbox{valign=t}{\includegraphics[width=0.35\linewidth]{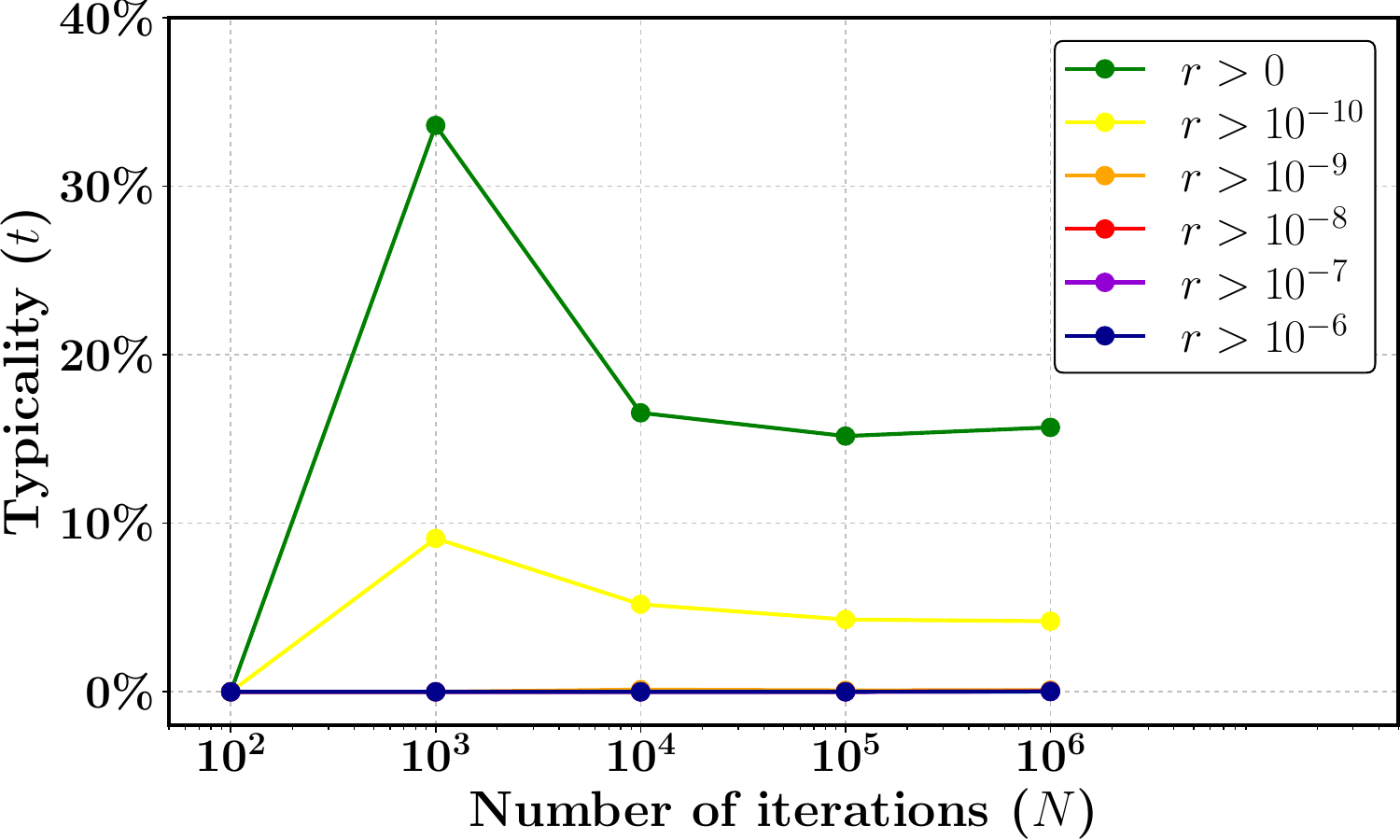}}
		\end{tabular}
		\caption{Plots of typicality for a scenario $(n=4,m=2,d=2,N)$ with varying numbers of iterations $N$ for $n$ mixed states and $m$ binary POVMs for (a) CLARABEL and (b) ECOS, and considering different thresholds for $r$.}
		\label{fig:thresh2}
	\end{figure*}
	
	The number $N=10^6$ is therefore chosen to ensure a decent statistical sample while still ensuring that the computational time for our analysis was feasible. We are aware that this number drastically changes in the literature of typicality (for instance, it ranges from $10^5$ to $10^{10}$ in Ref.~\cite{de2020strength}). This is evidently another parameter that can be adjusted manually to fit the description in the laboratory, and we discuss throughout the text whenever a quantity should be sensitive to this number. When we say that ``the typicality of contextuality for a scenario is $t$'', this should hence be understood as \emph{out of $N$ random samples of $n$ states and $2m$ effects for a Hilbert space with dimension $d$, $t\cdot N$ of them yielded $r>10^{-7}$ for the linear program in Ref.}~\cite{selby2024linear}.\\
	
	Finally, let us comment on the confidence level of our typicality results. Since our data constitutes a series of Bernoulli trials (i.e., success-failure experiments), we can employ binomial proportion confidence intervals to estimate, with some confidence level lower than 100\%, the interval range within which the true value of $t_{(n,m,d)}$ lies. For our analysis, we estimate the lower bound of the Wilson score interval~\cite{wilson} for the successful tries in our sampling for a confidence level of $99\%$. Given $t_{(n,m,d;N)}$ and $N_s$, the number of tries that actually yield valid typicality assessments (i.e., excluding trials that yield inaccurate or invalid contextuality assessments), the lower bound of the Wilson score interval is given by
	\begin{equation}
		\begin{split}
		t_{(n,m,d)}&\geq_{99\%}\frac{1}{1+\frac{z^2_{0.99}}{N_s}}\left(t_{(n,m,d;N)}+\frac{z_{0.99}^2}{2N_s}\right.\\
		&\left.-\frac{z_{0.99}}{2N_s}\sqrt{4N_s t_{(n,m,d;N)}(1-t_{(n,m,d;N)})+z_{0.99}^2}\right),
		\end{split}
	\end{equation}
	where $z_{0.99}$ is a function of the adopted confidence level (in this case, 99\%). In other words, when we say $t_{(n,m,d;N=10^6)}\approx100\%$ in this work, we mean that the Wilson's score interval criterion would tell us that typicality $t_{(n,m,d)}$ is at least 99.999\% with 99\% confidence, assuming that all $10^6$ trials were unproblematic.

	\section{Analytical proofs}\label{app:proofs}
	
	We begin by proving Proposition~\ref{prop2}.
	
	\upperbound*
	
	\begin{proof}
		Notice that for $n\leq d^2$ or $m\leq d^2/2$ (in the case of dichotomic measurements), $t_{(n,m,d)}=0$ by Lemma~\ref{lemma:TipiZeroForLI}, so we now demonstrate that $t_{(n,m,d)}<1$ for the remaining cases.
		
		Consider the GPT described by the full set of quantum states $\Omega$ for a quantum system of dimension $d$, and its dual $\mathcal{E}$ representing the full set of effects. Hereon, $\Omega$ will be treated as a compact convex body in its affine state space of real dimension $d^2-1$. Fix a sampling model that uniformly samples states in $\Omega$ and effects in $\mathcal{E}$, then for a randomly sampled finite set of effects $\mathcal{E}_m\subsetneq\mathcal{E}$ with $m>d^2/2$ the dual polytope $\mathcal{E}_m^*$ strictly contains $\Omega$ and is full-dimensional with nonzero probability.
		
		Since $\mathcal{E}_m^*$ generically strictly contains $\Omega$, it is possible to select $v_0\in\mathcal{E}_m^*\backslash\Omega$. Moreover, it is always possible to choose a full-dimensional simplex $\triangle\subsetneq\mathcal{E}_m^*$ containing $v_0$ as one of its vertices and the remaining vertices in $\Omega$, such that $\mathsf{Int}(\triangle)\cap\mathsf{Int}(\Omega)\neq\emptyset$. By construction, $\mathsf{Int}(\triangle)\cap\mathsf{Int}(\Omega)\subsetneq\triangle$ and $\mathcal{E}_m\subsetneq\triangle^*$, so any GPT fragment contained in these sets of states and effects will be simplex-embeddable.
		
		By  construction, there exists a boundary point $\omega\in\partial\Omega$ such that $\omega\in\mathsf{Int}(\triangle)$. Therefore, there is an $\varepsilon$-ball $B_\varepsilon(\omega)\subsetneq\triangle$ around $\omega$. It follows that, given $\mathsf{Vol}$ the Lebesgue volume in $\mathbb{R}^{d^2-1}$, $\Omega_\varepsilon:=B_\varepsilon(\omega)\cap\Omega$ contains interior points of $\Omega$ and thus will have positive Lebesgue volume. Therefore, there is a finite probability of sampling a state in this region given by
		\begin{equation}
			p:=\frac{\mathsf{Vol}[\Omega_\varepsilon]}{\mathsf{Vol}[\Omega]}>0,
		\end{equation}
		and, by independence, the probability of sampling $\Omega_n$ quantum states in this region is also positive and given by $p^n$. With probability $p^n>0$, then, we end up sampling a set of \sout{pure} quantum states that satisfies $\Omega_n\subset\triangle$, and therefore corresponds to a simplex-embeddable fragment. The typicality of contextuality is therefore bounded by 
		\begin{equation}\label{eq:bound}
			t_{(n,m,d)}\leq1-p^n<1.
		\end{equation}
		
		Similarly, fix a sampling model that uniformly samples \emph{pure} states in $\Omega$ and \emph{pure} effects in $\mathcal{E}$, then for $m>d^2/2$ there is a nonzero probability that a randomly sampled finite set of effects $\mathcal{E}_m$ will be such that $\mathcal{E}^*_m$ is full-dimensional and $\Omega\subsetneq\mathcal{E}_m^*$. We can therefore choose $v_0\in\mathcal{E}_m\backslash\Omega$ and construct a simplex $\triangle$ containing $v_0$ as one of its vertices and the remaining ones in the boundary $\partial\Omega$, such that dim$(\mathsf{AffHull}[\triangle])>2(d-1)$. This ensures that $\mathsf{Int}(\triangle)\cap\partial\Omega\neq\emptyset$, and moreover $\mathsf{Int}(\triangle)\cap\partial\Omega\subsetneq\triangle$.
		
		By construction, there is a boundary point $\omega\in\partial\Omega$ such that $\omega\in\mathsf{Int}(\triangle)$. Therefore, there is an $\varepsilon$-ball $B_\varepsilon(\omega)\subset\triangle$ around $\omega$. It follows that, given $\mu$ the Haar measure in $\partial\Omega$, $\Omega_\varepsilon:=B_\varepsilon(\omega)\cap\partial\Omega$ contains points of $\partial\Omega$ and thus will have positive Haar measure. Therefore, there is a finite probability of sampling a state in this region, given by
		\begin{equation}
			p:=\frac{\mu(\Omega_\varepsilon)}{\mu(\Omega)}>0,
		\end{equation}
		and by independence, the probability of sampling $\Omega_n$ pure states in this region is also positive and given by $p^n>0$, so the bound on typicality takes on the same form as Eq.~\ref{eq:bound}.
	\end{proof}
	
	Naturally, the upper bound in Eq.~\eqref{eq:bound} is not tight: there are infinitely many such simplices, and in this proof we restricted the scope of simplex-embeddings to those where the simplices have the same dimension as the sampled set of quantum states.
	
	We now proceed to prove Proposition~\ref{prop3}.
	
	\asymptotic*

	\begin{proof}
		Let us begin by proving Eq.~\eqref{eq:lim-effects}.
		
		Similarly to the previous proof, consider the GPT described by the full set of quantum states $\Omega$ for a quantum system of dimension $d$, and its dual $\mathcal{E}$ representing the full set of effects. Hereon, $\Omega$ and $\mathcal{E}$ will be treated as a compact convex bodies in their affine spaces of real dimension $d^2-1$. Notice that for any $d$, there is no simplex $\triangle$ and embeddings $\iota:\Omega\to\triangle$ and $\kappa:\mathcal{E}\to\triangle$ such that $(\iota(\Omega),\kappa(\mathcal{E})\subseteq(\triangle,\triangle^*)$.            
		
		Fix a sampling model that uniformly samples pure states in $\partial\Omega$ and pure effects in $\partial\mathcal{E}$. From Ref.~\cite{zhang2025reassessing}, any set $\Omega_n$ of $n>d^2$ randomly sampled states will be such that the fragment $(\Omega_n,\mathcal{E})$ does not admit of a simplex embedding. Assume that there is $\triangle$ and $\iota,\kappa$ such that the fragment $(\Omega_n,\mathcal{E}')$ admits of a simplex embedding for some convex hull $\mathcal{E}'$ of pure effects. Then the result in Ref.~\cite{zhang2025reassessing} imposes that $\mathcal{E}'\neq\mathcal{E}$, and since $\mathcal{E}$ is the set of all possible effects for this system, it must be that $\mathcal{E}'\subsetneq\mathcal{E}$, and therefore $\partial\mathcal{E}'\subsetneq\partial\mathcal{E}$. This means that $\partial\mathcal{E}'\cap\partial\mathcal{E}=\partial\mathcal{E}'$. It follows that, given $\mu$ the Haar measure in $\partial\mathcal{E}$, the likelihood of sampling an effect in $\partial\mathcal{E}'$ is
		\begin{equation}
			p_{\Omega_n,\triangle}:=\frac{\mu(\partial\mathcal{E}')}{\mu(\partial\mathcal{E})}<1.
		\end{equation}
		
		The probability of finding such an effect for any given $\Omega_n$ will be given by maximising $p_{(\Omega_n,\triangle)}$ over all possible simplex-embeddings, and then averaging that over all possible $\Omega_n$. That is,
		\begin{equation}
			p:=\int d\mu(\Omega_n)\max_{\triangle,\iota,\kappa}\frac{\mu(\partial\mathcal{E}')}{\mu(\partial\mathcal{E})}<1.
		\end{equation}
		The probability of independently sampling $m$ effects not in this setup is therefore given by
		\begin{equation}
			t_{(n,m,d)}=1-p^{m},
		\end{equation}
		and in the asymptotic limit,
		\begin{equation}
			\lim_{m\to\infty}t_{(n,m,d)}=1.
		\end{equation}
		
		The proof follows similarly for Eq.~\ref{eq:lim-states}. From the result in Ref.~\cite{zhang2025reassessing}, for any set $\mathcal{E}_m$ of dichotomic pure effects with $m\geq d^2/2$, linear dependences will be found among them and so $(\Omega,\mathcal{E}_m)$ will not admit of a simplex embedding. Assume therefore that there is a $\triangle$ and $\iota,\kappa$ such that $(\Omega',\mathcal{E}_m)$ admits of a simplex embedding for some convex hull $\Omega'$ of pure states. The result in Ref.~\cite{zhang2025reassessing} imposes that $\Omega'\neq\Omega$, and by the same reasoning as previously it follows that $\partial\Omega'\cap\partial\Omega=\partial\Omega'$, so the likelihood of sampling a state in $\partial\Omega'$ given the Haar measure $\mu$ in $\partial\Omega$ is
		\begin{equation}
			p_{\mathcal{E}_m,\triangle}:=\frac{\mu(\partial\Omega')}{\mu(\partial\Omega)}<1.
		\end{equation}
		Again we must maximise it over all possible simplex embeddings and then average over all possible sets $\mathcal{E}_m$, but this probability remains lower than one:
		\begin{equation}
			p:=\int d\mu(\mathcal{E}_m)\max_{\triangle,\iota,\kappa}\frac{\mu(\partial\Omega')}{\mu(\partial\Omega)}<1,
		\end{equation}
		so the probability of independently sampling $n$ states not in this setup is therefore given by
		\begin{equation}
			t_{(n,m,d)}:=1-p^n,
		\end{equation}
		so in the asymptotic limit
		\begin{equation}
			\lim_{n\to\infty}t_{(n,m,d)}=1.
		\end{equation}
	\end{proof}
	
	\section{Open-source repository}\label{app:repo}
	
	The repository with the codes and data generated in this work is available at Ref.~\cite{github}. Notice that the computations require the linear program introduced in Ref.~\cite{selby2024linear} and available at Ref.~\cite{cavalcanti-git}. Make sure that this program and its requirements are working properly before any attempt to run the functions in this project.
	
	The repository requires \texttt{qutip} version 5.2 or above. It samples random states as described in Sec.~\ref{sec:methods} and Appendix~\ref{app:qutip} through the functions \texttt{random\_density\_matrix}, \texttt{random\_effects}, \texttt{fixed\_effects} and \texttt{random\_unitary}. In particular, \texttt{random\_density\_matrix} will sample either (i) a pure state uniformly via Haar measure if \texttt{pure = True} is provided as an input or (ii) a full-rank mixed density operator by Ginibre sampling, employing the exclusion criterion if purity lies out of the interval determined by the inputs (\texttt{upperbound} and \texttt{lowerbound}). Similarly, \texttt{random\_effects} will sample either (i) $m$ projective measurements via Haar measure if \texttt{pure = True}; or (ii) $m$ POVMs with $k$ outputs (with $k=2$ by default) by Ginibre sampling, employing the exclusion criterion if the sharpness (Eq.~\ref{eq:sharpness}) lies out of the interval determined by the inputs (\texttt{upperbound} and \texttt{lowerbound}).
	
	We define a series of internal functions that build up to estimating typicality through parallel processing, culminating in the functions \texttt{Parallel\_Typicality} and \texttt{Parallel\_Typicality\_fixed}. \texttt{Parallel\_Typicality} takes in as parameters the number of preparations $n$; the number of measurements $m$; the dimension $d$ of the Hilbert space; the upper and lower bounds on purity of the sampled states; two Boolean variables, \texttt{pure\_preps} and \texttt{pure\_meas}, that regulate whether the states and effects are sampled from the Haar measure or through Ginibre sampling; the iteration number $N$; and an optional argument related to the number of workers employed in the multiprocessing. \texttt{Parallel\_Typicality\_fixed}, by its turn, does not take in $m$, $d$ or \texttt{pure\_meas}, since the measurements (and the dimension of the Hilbert space) are fixed in this case. By default, the number of workers is set through the function \texttt{multiprocessing.cpu\_count()}, but it should be adjusted manually if running on a shared system or in a machine with limited resources.
	
	The repository also provides the functions that generated the data for all plots in this paper. Reproducibility is subject to numerical instabilities and hardware peculiarities. We also include a function \texttt{Typicality\_POM} for the parity-oblivious multiplexing analysis made in Sec.~\ref{sec:POM}, and a function \texttt{wilson\_score\_interval} to estimate the lower bound of the Wilson score interval~\cite{wilson} for the computed frequencies.
	
	The core function of this repository is \texttt{Minimalpreps}. As arguments, it will take the same arguments as \texttt{Parallel\_Typicality}, with the exception of the number $n$ of preparations. It will then loop the function \texttt{Parallel\_Typicality} for an increasing number of preparations $n$, starting at $n=4$, until it returns a typicality greater than $99\%$. The function then returns the number $n$ for which this happened. Notice that if $m$ is set to be 1, the loop would never end (since any scenario would always be simplex-embeddable for a single measurement), and therefore the code will return a warning. \\
	
	Additional information, such as installation, licensing, and example usage, is provided in the README.  The data generated for this manuscript are also available in a compressed file in the repository.
\end{document}